\documentclass[final,letterpaper,11pt]{article}

\usepackage{amsmath, amssymb, amsthm, hyperref, setspace, fancyhdr, lastpage, graphicx, caption, stackengine, longtable ,bbm, indentfirst, enumitem, color,bm, mathtools, relsize}
\usepackage{psfrag,epsf}
\usepackage[toc,page]{appendix}
\usepackage{makecell}
\usepackage{multirow}
\usepackage{booktabs}
\usepackage{url} 
\usepackage{subcaption}
\usepackage{xcolor}
\usepackage{afterpage}

\newtheorem{cor}{Corollary}[section]
\newtheorem{prop}{Proposition}[section]

\usepackage{booktabs}
\usepackage[all]{xy}
\theoremstyle{plain}

\usepackage{algpseudocode}
\usepackage{algorithm}

\newtheorem{theorem}{Theorem}

\theoremstyle{definition}
\newtheorem{definition}[theorem]{Definition}

\newtheorem{remark}{Remark}

\newcommand{\rn}[1]{%
	\textup{\expandafter{\romannumeral#1}}%
}
\newcommand{\RN}[1]{%
	\textup{\uppercase\expandafter{\romannumeral#1}}%
}

\DeclarePairedDelimiter{\ceil}{\lceil}{\rceil}

\usepackage[normalem]{ulem}
\usepackage{color}

\newcommand{\stout}[1]{\ifmmode\text{\sout{\ensuremath{#1}}}\else\sout{#1}\fi}

\allowdisplaybreaks

\begin{document}


\title{\bf {Time Series Featurization \\ via Topological Data Analysis}}
\author{
		Kwangho Kim \thanks{Both authors contributed equally to this manuscript} \thanks{To whom correspondence should be addressed} \thanks{
		\{Department of Statistics \& Data Science, Machine Learning Department\}, Carnegie Mellon University, Pittsburgh, PA 15213. Email: \texttt{kwhanghk@stat.cmu.edu}} \hspace{1cm} 
		Jisu Kim \footnotemark[1] \thanks{
		Inria Saclay -- \^{I}le-de-France, Palaiseau, France.  Email: \texttt{jisu.kim@inria.fr}, URL: \texttt{http://www.stat.cmu.edu/~jisuk/}} \hspace{1cm} 
		Alessandro Rinaldo \thanks{Department of Statistics \& Data Science, Carnegie Mellon University, Pittsburgh, USA. Email: \texttt{arinaldo@cmu.edu}, URL: \texttt{http://www.stat.cmu.edu/~arinaldo/}}
}

\maketitle


\begin{abstract}
	We develop a novel algorithm for feature extraction in time series data by leveraging tools from topological data analysis. Our algorithm provides a simple, efficient way to successfully harness topological features of the attractor of the underlying dynamical system for an observed time series. The proposed methodology relies on the persistent landscapes and silhouette of the Rips complex obtained after a de-noising step based on principal components applied to a time-delayed embedding of a noisy, discrete time series sample. 
	We analyze the stability properties of the proposed approach and show that the resulting TDA-based features are robust to sampling noise. Experiments on synthetic and real-world data demonstrate the effectiveness of our approach. We expect our method to provide new insights on feature extraction from granular, noisy time series data.
\end{abstract}

\noindent%
{\it Keywords:}  stability properties,
topological features,
time-series feature extraction,
noisy time-series,
topological data analysis
\vfill

\newpage

\section{Introduction}

Time series analysis aims at predicting the evolution of a time varying phenomenon. However, time series data exhibit different characteristics than point cloud data obtained through random sampling. First of all, data points close in time are more likely to share common features than those which are further apart in time; this is referred to as serial \textit{autocorrelation}. Moreover, in time series forecasting we rarely know the underlying data generating distribution and the stochastic dependence across time associated to it, making applications of standard asymptotic results, such as the central limit theorem and the weak law of large numbers, problematic. Lastly, the granularity of the features are often very different from that of the outcome variable so simple summary statistics do not carry nuanced information about time-varying patterns.

These issues render the bulk of modern machine learning and feature engineering methods based on i.i.d. data inadequate, or at the very least cause a  degradation in their performance. Practitioners instead rely on stochastic models that have been developed based on the Ergodic theorem (e.g. ARIMA, SARIMA) or time-frequency analysis techniques (e.g. Fourier or Wavelet transform). More recently, due to the tremendous increase in the computing processing power and the developments of novel methodologies, researchers have started harnessing more innovational approaches once considered irrelevant to time series analysis (e.g., deep-learning architectures \cite{langkvist2014review, adhikari2013introductory, lyapplying2018, gullapalli2018learning}). One such approach that has attracted a growing interest is topological data analysis.

\textit{Topological data analysis} (TDA) is a recent and emerging field of data science that relies on topological and geometric tools to infer relevant features for possibly complex data, mainly using persistent homology \cite{chazal2017introduction}. It has provided alternative procedures to the traditional statistical and machine learning-based algorithms. TDA methods are now applied in many areas, ranging from document structure
representation \cite{zhu2013persistent} to cancer detection \cite{nanda2014simplicial}, task-optimal architecture of neural networks \cite{guss2018characterizing}, to name a few (see \cite{chazal2017introduction} for more examples). In particular, recently they have gained increasing importance for problems in signal processing \cite{perea2015sliding,chazal2013persistence,tralie2018quasi, seversky2016time,brown2009nonlinear,gamble2010exploring, pereira2015persistent, umeda2017time, venkataraman2016persistent, emrani2014persistent} or financial econometrics \cite{gidea2018topological, nobi2014correlation}.  We also refer to \cite{gholizadeh2018short} for a more comprehensive review of recent applications of TDA in time series analysis.


\textbf{Related Work.}
Even though there are a number of applications where researchers find harnessing persistent homology very effective in pattern recognition or signal analysis, only few of them provide a systematic way of constructing topological features for general time series with the underpinning theory. For example, \cite{perea2015sliding, tralie2018quasi} focus specifically on quantifying periodicity in particular types of signals. Studies such as \cite{gidea2018topological, seversky2016time,brown2009nonlinear,truong2017exploration,gamble2010exploring} have suggested that persistent homology can lead to a new type of features for time series or signal data, but their work mostly came away with simple visualizations or synthetic experiments. \cite{pereira2015persistent} used persistent homology to extract time series features for spatial clustering but their featurization relies on the partial information they subjectively choose.

In this paper, we propose a novel time series featurization method to construct informative and robust features based on topological data analysis, without any time-dependent structural, or probabilistic assumptions on the data generating process. Specifically, we provide a systematic procedure where one can successfully harness topological features of the attractor of the underlying dynamical system for an observed time series with favorable theoretical properties. Our proposed method could be an effective solution to feature engineering for granular time series data where a rich set of inherent patterns cannot be easily captured by traditional approaches. 

Our proposed method resembles work by \cite{umeda2017time, venkataraman2016persistent, emrani2014persistent} in that we rely on the delayed embedding and compute the persistent diagram to featurize topological properties. However, while all of the aforementioned methods simply used Betti numbers as their final TDA features, we use different theoretical tools for topological feature extraction. Importantly the theoretical underpinning for the aforementioned studies is very limited. Moreover, we incorporate an explicit denoising process to filter homological noise.

\textbf{Contributions.} The contributions of this paper are as follows.
(1) We develop a novel topological featurization algorithm based on homological persistence, after the state space reconstruction followed by a denoising process via principal component analysis. (2) We show our denoising process can preserve the significant topological features while reducing homological noise. (3) We provide a stability theorem which implies that our TDA features are robust to sampling noise. (4) Finally, we demonstrate the value of the proposed approach through real-world example where we perform not only a single time-point prediction but also complex pattern classifications.

\section{Topological Data Analysis} 
\label{sec:tda}
In this section we review basic concepts about persistence homology and stability theorems. We refer the reader to Appendix \ref{app:tda_details} and \cite{Hatcher2002,EdelsbrunnerH2010,ChazalCGGO2009,ChazalSGO2016} for further details and formal definitions. 
Throughout,  $\mathbb{X}$ will denote a subset of $\mathbb{R}^d$ and $\mathcal{X}$ a collection of points from it. We use the notation $\mathbb{B}_{\mathbb{R}^{d}}(x,r)$ (or $\mathbb{B}(x,r)$ when the space is clear from the context) for the open ball in $\mathbb{R}^{d}$ centered at $x$ with radius $r$.

\subsection{Simplicial complex, persistent homology, and diagrams}
\label{subsec:rips}

When inferring topological properties of $\mathbb{X}$ from its finite collection of samples $\mathcal{X}$, we rely on \textit{simplicial complex} $K$, a discrete structure built over the observed points to provide a topological approximation of the underlying space. The union of balls formed from the points of $\mathcal{X}$ has a nice combinatorial description. The \textit{\v{C}ech complex} of the set of balls $\{\mathbb{B}(x_i,r)\}_i$ is the simplicial complex whose vertices are the points $x_i \in \mathcal{X}$ and whose $k$-simplices correspond to $k+1$ balls with nonempty intersection. However, due to computational issues, the most common choice is the \emph{Vietoris-Rips complex} (or \emph{Rips complex}), where simplexes are built based on pairwise distances among its vertices.


\begin{definition} \label{def:background_rips}
	
	The \emph{Rips complex} $R_{\mathcal{X}}(r)$ is the simplicial complex defined as
	\begin{equation}
	R_{\mathcal{X}}(r):=\{ \sigma\subset\mathcal{X}:d(x_{i},x_{j})<2r,\forall x_{i},x_{j}\in\sigma\} .\label{eq:background_rips}
	\end{equation}
\end{definition}


Persistent homology is a multiscale approach to represent topological features.
A \emph{filtration} $\mathcal{F}$ is a collection of subspaces approximating the data points at different resolutions. 
\begin{definition}
	A \emph{filtration} $\mathcal{F}=\{ \mathcal{F}_{a}\}_{a\in\mathbb{R}}$
	is a collection of sets with $a\leq b$ implying  $\mathcal{F}_{a}\subset\mathcal{F}_{b}$.
	
\end{definition}

For a filtration $\mathcal{F}$ and for each $k\in\mathbb{N}_{0}=\mathbb{N}\cup\{0\}$, the
associated persistent homology $PH_{k}\mathcal{F}$ is a collection of $k$-th dimensional homology of each subset in $\mathcal{F}$. 

\begin{definition} 
	
	Let $\mathcal{F}$ be a filtration and let $k\in\mathbb{N}_{0}$.
	The associated $k$-th \emph{persistent homology} $PH_{k}\mathcal{F}$
	is a collection of vector spaces $\{ H_{k}\mathcal{F}_{a}\} _{a\in\mathbb{R}}$
	equipped with homomorphisms $\{ \imath_{k}^{a,b}\} _{a\leq b}$,
	where $H_{k}\mathcal{F}_{a}$ is the $k$-th dimensional homology of
	$\mathcal{F}_{a}$ and $\imath_{k}^{a,b}:H_{k}\mathcal{F}_{a}\to H_{k}\mathcal{F}_{b}$ is the homomorphism induced by the inclusion $\mathcal{F}_{a}\subset\mathcal{F}_{b}$. 
	
\end{definition}

For the $k$-th persistent homology
$PH_{k}\mathcal{F}$, the set of filtration levels at which a specific
homology appears is always an interval $[b,d)\subset[-\infty,\infty]$,
i.e. a specific homology is formed at some filtration value $b$
and dies when the inside hole is filled at another value $d>b$.

\begin{definition} 
	
	Let $\mathcal{F}$ be a filtration and let $k\in\mathbb{N}_{0}$.
	The corresponding $k$-th persistence diagram $Dgm_{k}(\mathcal{F})$ is a
	finite multiset of $(\mathbb{R}\cup\{\infty\})^{2}$, consisting
	of all pairs $(b,d)$ where $[b,d)$ is the interval of filtration
	values for which a specific homology appears in $PH_{k}\mathcal{F}$. $b$
	is called a birth time and $d$ is called a death time.
	
\end{definition}


One common metric to measure distance between two persistent diagrams is the \textit{bottleneck distance}.
\begin{definition}
	
	The bottleneck distance between the persistent homology of the filtrations
	$PH_{k}(f)$ and $PH_{k}(g)$ is defined by 
	\[
	d_{B}(PH_{k}(f),PH_{k}(g))=\inf\limits _{\gamma\in\Gamma}\sup\limits _{p \in Dgm_{k}(f)}\|p-\gamma(p)\|_{\infty},
	\]
	where the set $\Gamma$ consists of all the bijections $\gamma:Dgm_{k}(f)\cup Diag\rightarrow Dgm_{k}(g)\cup Diag$,
	and $Diag$ is the diagonal  $\{(x,x):\,x\in\mathbb{R}\}\subset\mathbb{R}^{2}$
	with infinite multiplicity. 
	
\end{definition}

The bottleneck distance imposes a metric structure on the space of persistence diagrams, which leads to the persistence stability theorem, i.e., small perturbations in the data implies at most small changes in the persistence diagrams in terms of the bottleneck distance (see Appendix \ref{app:tda_details} for more details).

\subsection{Landscapes and Silhouettes}

The space of persistence diagrams is a multiset, which makes it difficult to analyze and to feed as input to learning or statistical methods. Hence, it is useful to transform the persistent homology into a functional Hilbert space, where the analysis is easier and learning methods can be directly applied. Good examples include {persistent landscapes} \cite{bubenik2015statistical, bubenik2018persistence, bubenik2017persistence} and {silhouettes} \cite{chazal2015subsampling, chazal2014stochastic}, both of which are real-valued functions that further summarize the crucial information contained in a persistence diagram. We briefly introduce the two functions.

\textbf{Landscapes.} 
Given a persistence diagram $Dgm_{k}(\mathcal{F})$, we first define a set of functions $t \in \mathbb{R}_{+0} \mapsto \Lambda_p(t)$ for each birth-death pair $p=(b, d)$ in the persistence diagram as follows:
\begin{align} \label{def:tenting-function}
\Lambda_p(t) = 
\begin{cases}
t-b, & t \in [b, \frac{b+d}{2}]\\
d-t, & t \in (\frac{b+d}{2}, d ]\\
0, & \text{otherwise}
\end{cases}.
\end{align}
For each birth-death pair $p$, $\Lambda_p(\cdot)$ is piecewise linear. Then the \textit{persistence landscape} $\lambda$ of the persistent diagram $Dgm_{k}(\mathcal{F})$ is defined by the sequence of functions $\{\lambda_j\}_{j\in \mathbb{N}}$, where
\begin{equation} \label{def:landscape}
\lambda(j,t) = \underset{p\in Dgm_{k}(\mathcal{F})}{jmax} \Lambda_p(t), \quad  t\in[0,d_{max}], j\in \mathbb{N},
\end{equation}
and $jmax$ is the $j$-th largest value in the set. Hence, the persistence landscape is a function $\lambda : \mathbb{N} \times \mathbb{R} \rightarrow [0, d_{max}]$ for some large enough $d_{max}$. 

\textbf{Silhouette.} Let $Dgm_{k}(\mathcal{F})$ be a persistent diagram that contains $N$ off diagonal birth-death pairs $\{(b_j, d_j) \}_{j=1}^N$. We define the \textit{silhouette} as the following power-weighted function,
\begin{align} \label{def:silhouette}
\phi^{(q)}(t) = \frac{\sum_{j=1}^N \vert d_j - b_j \vert^q \Lambda_j(t) }{\sum_{j=1}^N \vert d_j - b_j \vert^q}.
\end{align}
for $0 < q < \infty$. The value $q$ can be thought of as controlling the trade-off parameter between uniformly treating all pairs in the persistence diagram and considering only the most persistent pairs. In fact when $q \rightarrow \infty$, silhouettes converge to the first order landscapes ($k=1$).		

Note that both landscape and silhouette functions can be evaluated over $\mathbb{R}$, and are easy to compute. Many recent studies including previously listed ones have revealed that these functional summaries of the persistence module have favorable theoretical properties, and can also be easily averaged and used for subsequent statistics and machine learning modeling \cite{chazal2014stochastic, bubenik2018persistence, bubenik2017persistence}.

\section{Homological noise reduction via PCA}
\label{sec:pca-denoise}

In this section we demonstrate how principal component analysis (PCA) can lead to effective methods for filtering out topological noise while preserving  significant topological features in the data. In detail, suppose we are observe
a finite point cloud that has been contaminated with additive errors. Altogether, these
error terms may bring in a significant amount of topological noise. A standard way to denoise the noises is to resort to
PCA.

As we will see below, the first $l$ principal components (PCs) preserves the topological features when the signal data points before contamination lie near a $l$-dimensional
linear space. To begin with, we may assume that the the signal data points belong to some $l$-dimensional linear space. Then the bottleneck distance
between the persistent homology of the signal data and the persistent
homology of the first $l$ leading PCs is roughly bounded by the magnitude the contamination, i.e. of the additive error. In particular, when there is no contamination, then the persistent
homology is preserved perfectly well after applying PCA.

\begin{prop} \label{prop:pca_stability_persistent_homology_linear}
	Let $\mathbb{X}=\{\mathbb{X}_{1},\ldots,\mathbb{X}_{n}\}\subset\mathbb{R}^{m}$
	and $X=\{X_{1},\ldots,X_{n}\}\subset\mathbb{R}^{m}$ be two point
	clouds of $\mathbb{R}^{m}$ with same size. Let $\mathbb{V}\subset\mathbb{R}^{m}$
	be an $l$-dimensional linear subspace passing $0$ satisfying $\mathbb{X}\subset\mathbb{V}$.
	Let $X^{l}$ be the outcome of the PCA on $X$. Let $D_{\mathbb{X}}$
	and $D_{X^{l}}$ denote the persistence diagrams of $\{R_{\mathbb{X}}(r)\}_{r\in\mathbb{R}}$
	and $\{R_{X^{l}}(r)\}_{r\in\mathbb{R}}$, respectively. Suppose $\frac{1}{n}\mathbb{X}^{\top}\mathbb{X}$
	has $\tilde{l}\leq l$ positive eigenvalues $\lambda_{1}\geq\cdots\geq\lambda_{\tilde{l}}>0$.
	Then the bottleneck distance between $D_{\mathbb{X}}$ and $D_{X^{l}}$
	is bounded as 
	\[
	d_{B}(D_{\mathbb{X}},D_{X^{l}})\leq\sup_{i}\left\Vert \mathbb{X}_{i}-X_{i}\right\Vert _{2}\left(1+\frac{2\sup_{i}\left\Vert \mathbb{X}_{i}\right\Vert _{2}\left(\sup_{i}\left\Vert \mathbb{X}_{i}-X_{i}\right\Vert _{2}+2\sup_{i}\left\Vert \mathbb{X}_{i}\right\Vert _{2}\right)}{\lambda_{\tilde{l}}}\right).
	\]
	In particular, when $\mathbb{X}=X$, then the persistent homology
	of $\mathbb{X}$ is preserved after applying PCA, i.e.
	\[
	d_{B}(D_{\mathbb{X}},D_{X^{l}})=0.
	\]
	
\end{prop}

More generally, suppose the signal data points lie near some
$l$-dimensional linear space. Then the bottleneck distance between
the persistent homology of the signal data and the persistent homology
of the leading PCs is bounded by an amount proportional to the level of contamination 
plus the distance from the signal data to the approximating linear space.

\begin{cor} \label{cor:pca_stability_persistent_homology_general}
	Let $\mathbb{X}=\{\mathbb{X}_{1},\ldots,\mathbb{X}_{n}\}\subset\mathbb{R}^{m}$
	and $X=\{X_{1},\ldots,X_{n}\}\subset\mathbb{R}^{m}$ be two point
	clouds of $\mathbb{R}^{m}$ with same size. Let $\mathbb{V}\subset\mathbb{R}^{m}$
	be an $l$-dimensional linear subspace passing $0$, and let $\mathbb{X}_{\mathbb{V}}$
	be the projections of $\mathbb{X}$ on $\mathbb{V}$, i.e. $(\mathbb{X}_{\mathbb{V}})_{i}=\Pi_{\mathbb{V}}(\mathbb{X}_{i})$.
	Let $X^{l}$ be the outcome of the PCA on $X$. Let $D_{\mathbb{X}}$
	and $D_{X^{l}}$ denote the persistence diagrams of $\{R_{\mathbb{X}}(r)\}_{r\in\mathbb{R}}$
	and $\{R_{X^{l}}(r)\}_{r\in\mathbb{R}}$, respectively. Suppose $\frac{1}{n}\mathbb{X}_{\mathbb{V}}^{\top}\mathbb{X}_{\mathbb{V}}$
	has $\tilde{l}\leq l$ positive eigenvalues $\lambda_{1}\geq\cdots\geq\lambda_{\tilde{l}}>0$.
	Then, 
	\begin{align*}
	d_{B}(D_{\mathbb{X}},D_{X^{l}}) & \leq d_{H}(\mathbb{X},\mathbb{V})+\left(d_{H}(\mathbb{X},\mathbb{V})+\sup_{i}\left\Vert \mathbb{X}_{i}-X_{i}\right\Vert _{2}\right)\\
	& \ \times\left(1+\frac{2\sup_{i}\left\Vert \mathbb{X}_{i}\right\Vert _{2}\left(d_{H}(\mathbb{X},\mathbb{V})+\sup_{i}\left\Vert \mathbb{X}_{i}-X_{i}\right\Vert _{2}+\sup_{i}\left\Vert \mathbb{X}_{i}\right\Vert _{2}\right)}{\lambda_{\tilde{l}}}\right).
	\end{align*}
	
\end{cor}

Hence, application of PCA in the scenario will preserve the topological features that are well captured by the approximating  $l$-dimensional linear
space.

We now discuss the impact of PCA on the topological noise arising from the contamination process.
Since such error contamination is the result of adding full-dimensional noise,
the corresponding topological noises are also full-dimensional.
Suppose first that the topological noise is topologically a sphere of dimension
$l$ or higher. When projected to a $l$-dimensional linear
space, this sphere reduces to a disk of dimension $l$.  Hence, this topological noise is eliminated after the PCA. See Figure \ref{fig:pca_homology_high} in Appendix \ref{sec:appendix-PCA-example} for an illustrative example. When the topological noise is of
dimension smaller than $l$, it is unlikely that it will be aligned perfectly with the linear subspace where the data
are projected by PCA. As a result, only the topological topological noise along
the orthogonal direction of the linear subspace is reduced. In particular,
if the topological noise is aligned with the orthogonal direction
of the linear subspace, then it will be eliminated entirely. See Figure \ref{fig:pca_homology} in Appendix \ref{sec:appendix-PCA-example} for an illustrative example.

In summary, applying PCA will reduce the impact of the topological noise in two ways: it will  reduce high-dimensional topological noise and, at the same time, it will eliminate
lower-dimensional topological  noise along the orthogonal direction
of the linear space where the data are projected. Meanwhile,
applying PCA still preserves the topological features that are lying
on or near a linear space of appropriate dimension. 

\section{Methodology}
\subsection{State space reconstruction via time-delayed embedding}

A time series can be considered a series of projections of the observed states from a \textit{dynamical system} which is a rule for time evolution on a state space. Since it is in general very difficult to reconstruct the state space and transition rules of the dynamical system from observed time series data without a priori knowledge \cite{tong1990non, eckmann1992fundamental}, we typically rely on \textit{attractors}, each of which is a set of numerical values toward which the given dynamical system eventually evolves over time, passing through various transitions determined by the system \cite{tong1990non, packard1980geometry, eckmann1992fundamental}. Since we need infinitely many points to construct any attractor, we instead use a \textit{quasi-attractor} that can be derived from a finite sample set. One of the most well-known approaches is time-delayed embedding including Takens' delay embedding \cite{takens1981detecting}. The time-delayed embedding theorems have been applied to a wide range of scientific disciplines to study chaotic or noisy time series (to name a few, \cite{torku2016takens, casdagli1989nonlinear, huang2014manifold, basharat2009time, theiler1992testing, kostelich1993noise, jayawardena1994analysis,sugihara1990nonlinear}). We refer to Appendix \ref{app:dynamic-systems} for more information on dynamical systems.

We adopt the time-delayed sliding window embedding following the work of \cite{perea2015sliding,tralie2018quasi}, in which we slightly tailor their original approach. For now, we only consider one dimensional time series data. Let $f(t)$ be a function defined over the non-negative reals $\mathbb{R}_{+0}$. Since time series featurization will be done only upon finitely bounded time interval, it is sufficient to consider a continuous function $f(t)$ on an interval $t\in[0,T]$ with $0<T<\infty$. First, let 
$ x = \{ x_0, x_1, ... , x_N \} $ be a sequence of equi-interval samples from the function $f$. Then, let the \textit{sliding window mapping} $SW_{m,\tau}f: R \rightarrow R^m$ be
\[
SW_{m,\tau}f(t) \coloneqq \big[f\left(t-(m-1)\tau\right),..., f(t-\tau), f(t)\big]^\top.
\]

In other words, for a fixed function $f$ the map $t \mapsto SW_{m,\tau}f(t)$ generates the $m$ most recent equally-spaced points up to time $t$ with a predefined gap $\tau \in \mathbb{N}$. In fact, parameters that appear in the above definition characterize the transformation according to the Takens' embedding theorem; $\tau$ will be the delay parameter and $m$ will be our embedding dimension. Now, given $\tau, m$ we construct the trajectory matrix $X \in \mathbb{R}^{\{N-(m-1)\tau\}\times m}$ as follows:
\begin{equation} \label{trajectory_mat}
\small
X = 
\begin{bmatrix}
SW_{m,\tau}f\left((m-1)\tau\right)^\top  \\
SW_{m,\tau}f\left(1+(m-1)\tau\right)^\top \\
\vdots  \\
SW_{m,\tau}f\left(T\right)^\top
\end{bmatrix} 
= 
\begin{bmatrix}
X_{(m-1)\tau} \\
X_{1+(m-1)\tau} \\
\vdots  \\
X_{N}  
\end{bmatrix} 
=
\begin{bmatrix}
x_{0} & x_{\tau} & \dots  & x_{(m-1)\tau} \\
x_{1} & x_{1+\tau} &  \dots  & x_{1+(m-1)\tau} \\
\vdots & \vdots &  \ddots & \vdots \\
x_{N-(m-1)\tau} & x_{N-(m-2)\tau} &  \dots  & x_{N}.
\end{bmatrix}
\normalsize
\end{equation}

In the above construction, we set $f(0)=x_1, ..., f(T)=x_N$. 
Our quasi-attractor is then the topological pattern traced out by $X$ in $\mathbb{R}^m$. By the  generalized Takens' theorem \cite{robinson2005topological} there always exists a pair $(\tau,m)$ such that the vectors generated via the transformation described in (\ref{trajectory_mat}) are on a manifold topologically equivalent to the attractor of the original dynamical system for $f$ \cite{torku2016takens}. In fact, the last display in (\ref{trajectory_mat}) is a typical expression for Takens' delay embedding transformation \cite{venkataraman2016persistent, emrani2014persistent, pereira2015persistent, seversky2016time}. We will later see this construction is beneficial for the analysis of our stability results in section \ref{sec:stability-thm}.

The transformation procedure described above requires the choice of the embedding dimension, $m$, and the time delay, $\tau$. There is no generic optimal method for choosing these parameters without any structural assumptions \cite{small2003optimal,strozzi2002application}, and practitioners rather rely on heuristic methods. $m$ must be sufficiently large to satisfy $m > 2d_0$ but we never know $d_0$ in reality. To determine $d_0$, we count linearly independent vectors in the singular value decomposition of $X$ following the work of \cite{palu1992singular, torku2016takens}, and pick a set of consecutive numbers equal to or greater than that. As for $\tau$, this parameter can be set using  autocorrelation or mutual information  \cite{small2003optimal}. Nonetheless $\tau$ is more often chosen in a qualitative manner. In fact, it is possible that several $(m, \tau)$'s yield the same topological equivalence. Thus in the simulation, we use a grid of values for $(m, \tau)$ and report the best result.

\subsection{Feature extraction via persistent homology} \label{sec:featurization-theory}

In the previous section, we have shown how to construct a vector space that is on a manifold topologically equivalent to an attractor of the original dynamical system. In this section, we develop a methodology with which we can translate the valuable information about the quasi-attractor into vector-valued features based on persistent homology.

TDA is known to be very useful to analyze point cloud data \cite{carlsson2014topological}. As seen in the previous subsection, we need to analyze the topological structure (the quasi-attractor) of the point cloud in the vector space formed by the trajectory matrix in (\ref{trajectory_mat}).
To this end, we form the Rips complex and compute the corresponding persistent diagram. Then we use vectorized versions of the persistent landscapes (\ref{def:landscape}) and silhouettes (\ref{def:silhouette}) as our final TDA features. 

Before computing the persistent homology, we may want to perform a denoising process to reduce insignificant (noise-like) homological features. As shown in Section \ref{sec:pca-denoise}, this can be done via PCA so that our point cloud in $\mathbb{R}^m$ is transformed into $\mathbb{R}^l$ where $l \ll m$. Then we construct the Rips complex in this new space to evaluate the persistent homology.


\subsection{Algorithm}
In this section, we present our algorithm to implement the proposed featurization method in a step-by-step manners. First, we form a point cloud  in $\mathbb{R}^m$ from the observed time series data via (\ref{trajectory_mat}) with our choice of $(m, \tau)$. After the point cloud embedding, we perform PCA on $X$ to reduce the dimension from $\mathbb{R}^m$ to $\mathbb{R}^l$, and obtain $X^l$. Next, we construct the Rips filtration $R_{X^{l}}=\{ R_{X^{l}}(r)\} _{r\in\mathbb{R}}$ and compute the corresponding persistence diagram $Dgm(X^l)$, and obtain the landscapes or silhouette functions accordingly. Finally, we vectorize each of the landscapes or silhouette functions so that we can use them as inputs to machine learning methods. For the vectorization step, we set a resolution parameter $\kappa$ and then take $\ceil{d_{max}/\kappa}$ evenly spaced samples from each landscape or silhouette function. If we let $n_{dim} \in \mathbb{N}^{+0}$ denote our maximum homological dimension, then in the end we will have $(n_{dim}+1)$ different vectors of size $\ceil{d_{max}/\kappa}$ as our final TDA features (or one long vector by concatenating all the $(n_{dim}+1)$ vectors). We detail the entire procedure in Algorithm \ref{algorithm:TDA-featurization}.

\begin{algorithm}
	\caption{ Time Series Featurization via TDA} \label{algorithm:TDA-featurization}
	\textbf{Input:} Time series sequence $x = \{ x_1, x_2, ... , x_N \}$ with a fixed sampling rate
	
	\begin{enumerate}
		\item Construct the point cloud $X \in \mathbb{R}^m$ via the trajectory matrix (\ref{trajectory_mat}) with $m, \tau$
		\item Perform PCA on $X$ and obtain $X^l \in \mathbb{R}^l$ ($l \ll m$)
		\item Construct the Rips complex $R_{X^l}$ and compute the persistence diagram $Dgm(X^l)$ with $d_{max}$
		\item From $Dgm(X^l)$, compute the landscape $\lambda(k,t)$ (\ref{def:landscape}) or the silhouette $\phi^{(p)}(t)$ (\ref{def:silhouette})
		\item Vectorize $\lambda(k,t)$ or $\phi^{(p)}(t)$ at resolution $\kappa$
	\end{enumerate}
	
	\textbf{Output:} $(n_{dim}+1)$ different vectors of size $\ceil{d_{max}/\kappa}$
\end{algorithm}

\begin{remark}
	Another crucial reason for performing PCA in Step 2 of Algorithm \ref{algorithm:TDA-featurization} is to reduce the computational cost of constructing the simplicial complex. Of course, other dimensionality reduction techniques may be considered in place of PCA, such as LLE, ISOMAP, Kernel PCA or diffusion maps, which may better adapt to the intrinsic geometric structure of the data.
\end{remark}

Even though our algorithm is basically for one dimensional time series, one can extend our method to the multivariate case by simply concatenating all the TDA features,  each of which  obtained from different time series. Finally, we remark that, due to its construction,  the resulting topological feature may be sparse, in the sense of contain many zeros. 



\section{Stability Analysis} \label{sec:stability-thm}


In this section, we demonstrate that the proposed method is robust against noise and finite sampling errors. Suppose there is a true signal function $f:[0,T]\to\mathbb{R}$ which is corrupted with additive noise $\epsilon:[0,T]\to\mathbb{R}$, and we end up observing $N$ sample points $x$ from $f+\epsilon$ uniformly over $[0,T]$. Our goal is to infer topological features of the Takens embedding of $f$ only using $x$. In what follows, we provide a novel stability theorem which guarantees that the persistence diagram computed from Algorithm \ref{algorithm:TDA-featurization} does not arbitrary get farther away from our target (true) persistence diagram.


\begin{prop} \label{prop:analysis_stability_persistent_homology}
	Let $f:[0,T]\to\mathbb{R}$ be a Lipschitz function with Lipschitz constant $L_{f}$. Suppose $x$ is a time-series of $N$-samples from $f+\epsilon$ within the time interval $[0,T]$. Let $\mathbb{X}:=\{SW_{m,\tau}f(t):\,t\in [0,T]\}$ be the Takens embedding of $f$. Let $X$ be a point cloud constructed via \eqref{trajectory_mat} with fixed $m \in \mathbb{R}$ and $X^{l}$ be the outcome of the PCA on $X$. Let $\mathbb{X}^{N}:=\{SW_{m,\tau}f(t):\,t=(0+(m-1)\tau),(1+(m-1)\tau),\ldots,T\}$, and suppose $\mathbb{X}^{N}$ is contained in $l$-dimensional linear space and  $\frac{(\mathbb{X}^{N})^{\top}\mathbb{X}^{N}}{N-(m-1)\tau}$ has $\tilde{l}\leq l$ positive eigenvalues
	$\lambda_{1}\geq\ldots\geq\lambda_{\tilde{l}}>0$ when $\mathbb{X}^{N}$ is understood as a matrix. Let $D_{X}, D_{f}$ denote the persistence diagrams of $\{R_{X^{l}}(r)\}_{r\in\mathbb{R}}$ and $\{R_{\mathbb{X}}(r)\}_{r\in\mathbb{R}}$, respectively. Then, 
	\begin{equation}
	d_B(D_{X}, D_{f}) \leq \sqrt{m}\left\Vert \epsilon\right\Vert _{\infty}+\frac{2m^{\frac{3}{2}}L_{f}T\left\Vert \epsilon\right\Vert _{\infty}(2L_{f}T+\left\Vert \epsilon\right\Vert _{\infty})}{\lambda_{\tilde{l}}}+\frac{\sqrt{m}L_{f}T}{N}.\label{eq:analysis_stability_persistent_homology}
	\end{equation}
\end{prop}




From the above result, we obtain an analogous stability result for landscapes and silhouettes, both of which are part of the final output of our method. 

\begin{theorem} \label{thm:analysis_stability_landscape_silhouette}
	Let $D_{X}, D_{f}$ denote persistence diagrams from Proposition \ref{prop:analysis_stability_persistent_homology}. Let $\lambda_{k,X}$, $\lambda_{k,f}$ be the landscape functions from $D_{X}$ and $D_{f}$. Then 
	\begin{equation}
	\Vert \lambda_{k,X} - \lambda_{k,f} \Vert_\infty  \leq \sqrt{m}\left\Vert \epsilon\right\Vert _{\infty}+\frac{2m^{\frac{3}{2}}L_{f}T\left\Vert \epsilon\right\Vert _{\infty}(2L_{f}T+\left\Vert \epsilon\right\Vert _{\infty})}{\lambda_{\tilde{l}}}+\frac{\sqrt{m}L_{f}T}{N}.\label{eq:analysis_stability_landscape}
	\end{equation}
	Also, let $\phi^{(p)}_{X}$, $\phi^{(p)}_{f}$ be the silhouette functions from $D_{X}$ and $D_{f}$. Then 
	\begin{equation}
	\Vert \phi^{(p)}_{X} - \phi^{(p)}_{f} \Vert_\infty  \leq \sqrt{m}\left\Vert \epsilon\right\Vert _{\infty}+\frac{2m^{\frac{3}{2}}L_{f}T\left\Vert \epsilon\right\Vert _{\infty}(2L_{f}T+\left\Vert \epsilon\right\Vert _{\infty})}{\lambda_{\tilde{l}}}+\frac{\sqrt{m}L_{f}T}{N}.\label{eq:analysis_stability_silhouette}
	\end{equation}
\end{theorem}

The proof can be found in the Appendix \ref{app:proof-stability}. When inspecting the various terms in the upper bounds in \eqref{eq:analysis_stability_persistent_homology}, \eqref{eq:analysis_stability_landscape}, \eqref{eq:analysis_stability_silhouette}, we see that the first and second terms are stem from the noise $\epsilon$ and the PCA process respectively, and are proportional to $\left\Vert \epsilon\right\Vert_{\infty}$ or $\left\Vert \epsilon\right\Vert^2_{\infty}$. The last term is due the facts that the observations are observed only at discrete times, and is inversely proportional to $N$. Hence, when the noise $\epsilon$ is small and the size of our data is large enough, then we are guaranteed to closely approximate the topological information of the signal from our method.

\section{Experiments}
In this section, we demonstrate the effectiveness of our method through experiments. We use synthetic data sampled from different stochastic models and Bitcoin historical data for real-world example.

\textbf{Stochastic model classification.}
We aim to classify types of stochastic model based on an observed sequence of samples. We employ three different types of stochastic model as follows: $\text{ARIMA$_{1,1,2}$: } {x}_t= x_{t-1}  +  \phi_1\Delta_{x_{t-1}} - \theta_1e_{t-1} - \theta_2e_{t-2}$, $\text{Composite Sinusoidal: } asin(x_t)sin(bx_t)+c\cos(x_t+d)$ where $a,b,c,d \sim Uniform[1.45,1.55]$, $\text{Ornstein-Uhlenbeck: }  dx_{t}=\theta (\mu -x_{t})dt+\sigma dW_{t}$, where $W_{t}$ denotes the Wiener process. We use $(\phi_1, \theta_1, \theta_2) = (0.4, 0.2, 0.1)$,  $(\theta,\mu,\sigma)=(-0.5,0,0.5)$. We generate 1,000 samples each with $N=250$, and then extract topological features using Algorithm \ref{algorithm:TDA-featurization} with $m=25, \tau=5, d_{max}=1$, using 1st-order landscapes up to the 2nd order homology ($n_{dim}=2$). We use PCA so that $X^l \in \mathbb{R}^3$. Inspired by \cite{DBLP2018}, we feed each input vector into fully-connected two-layer perceptron (2LP) with 100 hidden units, and two-layer convolutional neural network (CNN) with a $1\times10$ filter and 100 hidden units. For each method, we classify type of the original stochastic model in which the sampled sequence originates, with and without TDA features. To compute the persistence diagram we use the R package \texttt{TDA}  \cite{fasy2014introduction}.

The result in Table \ref{tbl:sim-stochastic} shows the proposed method lowers error rates up to roughly $4\% \sim7 \%$ even in this simplified setting. One possible reason why the TDA features can provide additional explanatory power is illustrated in Figure \ref{fig:sim-stochastic} where we present TDA features of 9 sample sequences (3 for each model) computed by Algorithm \ref{algorithm:TDA-featurization}; compared to others, samples from the composite sinusoidal model appear to exhibit strong $H_1$ features with substantial persistence.

\begin{minipage}[t]{\textwidth}
	\begin{minipage}[b]{0.72\textwidth}
		\centering
		\includegraphics[width=0.99\columnwidth]{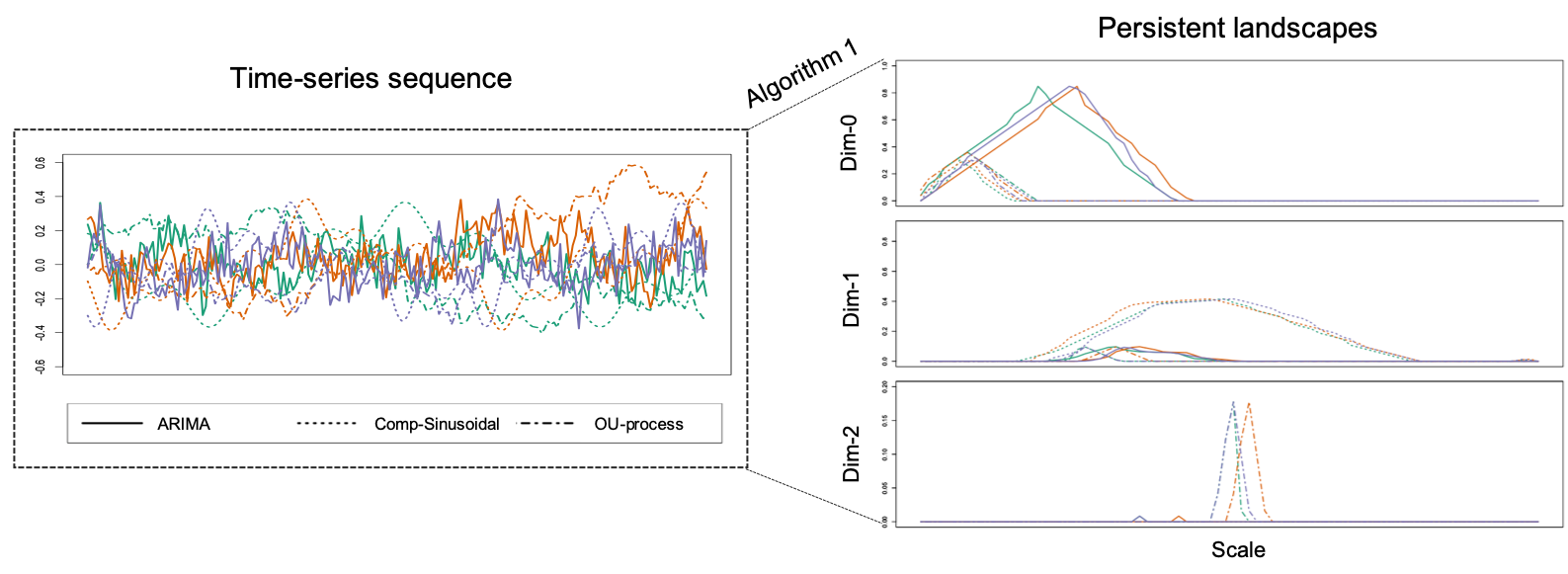}
		\captionof{figure}{TDA features for sampled stochastic sequences.}
		\label{fig:sim-stochastic}
	\end{minipage}
	\hfill
	\begin{minipage}[b]{0.27\textwidth}
		\centering
		\begin{tabular}{cc} \hline
			Method & Error rate \\ \hline
			\midrule
			2LP & 48.3\% \\
			CNN & 44.9\% \\
			TDA+2LP & 41.2\% \\
			TDA+CNN & 40.5\% \\ \hline
		\end{tabular}
		\captionof{table}{Error rates w/ and w/o TDA features.}
		\label{tbl:sim-stochastic}
	\end{minipage}
\end{minipage}

\textbf{Bitcoin price dataset.}
In this experiment, we seek to accurately identify and predict time patterns from the time series of Bitcoin prices, which are known to be very volatile and noisy \cite{segerestedt2018accuracy}. We prepare Bitcoin price data $\{y_t\}_t$ during the period from 01/01/2016 to 07/31/2018, at three different sampling rates - daily (D), hourly (H), and minute-by-minute (M) basis. We use the return over a single period $r_t= (p_{t}-p_{t-1})/p_{t-1}$ based on two consecutive closing prices $p_t$, $p_{t-1}$.  We carry out two types of tasks. The first one is to predict the next return, i.e. to predict $y_{t+1}$ using all the information up to $t$ where the prediction error is measure by root-mean-square error (RMSE). The second task is that of  price pattern classification over the next $k$ time points $[t,t+k]$ where we use four types of pattern as follows. $(P1)$: price movement in the next (single) time point. $(P2), (P3)$: determine if a rare (defined by $(1-\alpha)$-quantile) jump (P2) or drop (P3) occurs in $[t,t+k]$. $(P4)$: slope of moving average price movement on average in $[t,t+k]$.(P1), (P4) have three classes (\{up,down,neutral\} regimes), and (P2), (P3) have two classes (exist or not), whose threshold parameter is to be determined to have almost the same number of samples in each class. For simulation we use $k=6$.

We employ six different prediction models: non-parametric superlearner ensemble (NP) \cite{van2007super} with and without TDA features, 2-layer CNN with and without TDA features as previously, and traditional ARIMA and ARIMAX. When we do not use TDA features, we simply feed the time series sequence itself as an input. Following \cite{campbell1993trading}, we use trading volume as an additional covariate and compare the result with univariate case. We use grid search to find out the best topological parameters $\tau, m, d_{max}$, and set $n_{dim}=2$ and $X^{l} \in \mathbb{R}^3$. The results for the point prediction and the pattern classification are presented in Figure \ref{fig:sim-bitcoin-rmse} and Figure \ref{fig:sim-bitcoin-pattern} respectively. In Figure \ref{fig:sim-bitcoin-rmse}, Uni/S, Co/S on the x-axis mean with and without covariate at sampling rate $S$, where $S \in \{D,H,M\}$. Likewise in Figure \ref{fig:sim-bitcoin-pattern}, $n/S$ on the x-axis means classification for the pattern $P_n$ at sampling rate $S$. As we can see, utilizing TDA features has increased both the prediction and classification accuracy. In our experiments, CNN with TDA features delivers the best performance. We refer to Section \ref{app:sim-results} in the Appendix for more detailed information about the experiment.

\begin{minipage}[t]{\textwidth}
	\begin{minipage}[b]{0.33\textwidth}
		\centering
		\includegraphics[width=0.99\columnwidth]{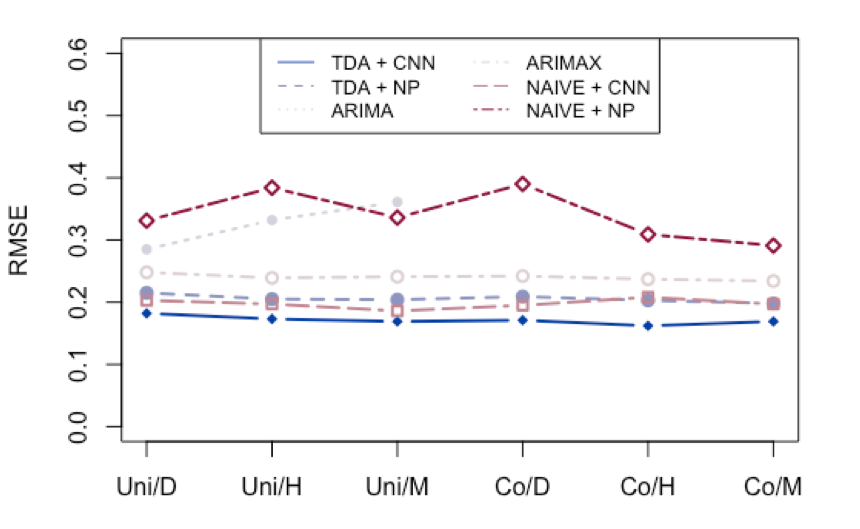}
		\captionof{figure}{RMSE for $\widehat{y}_{t+1}$}
		\label{fig:sim-bitcoin-rmse}
	\end{minipage}
	\hfill
	\begin{minipage}[b]{0.66\textwidth}
		\centering
		\includegraphics[width=0.99\columnwidth]{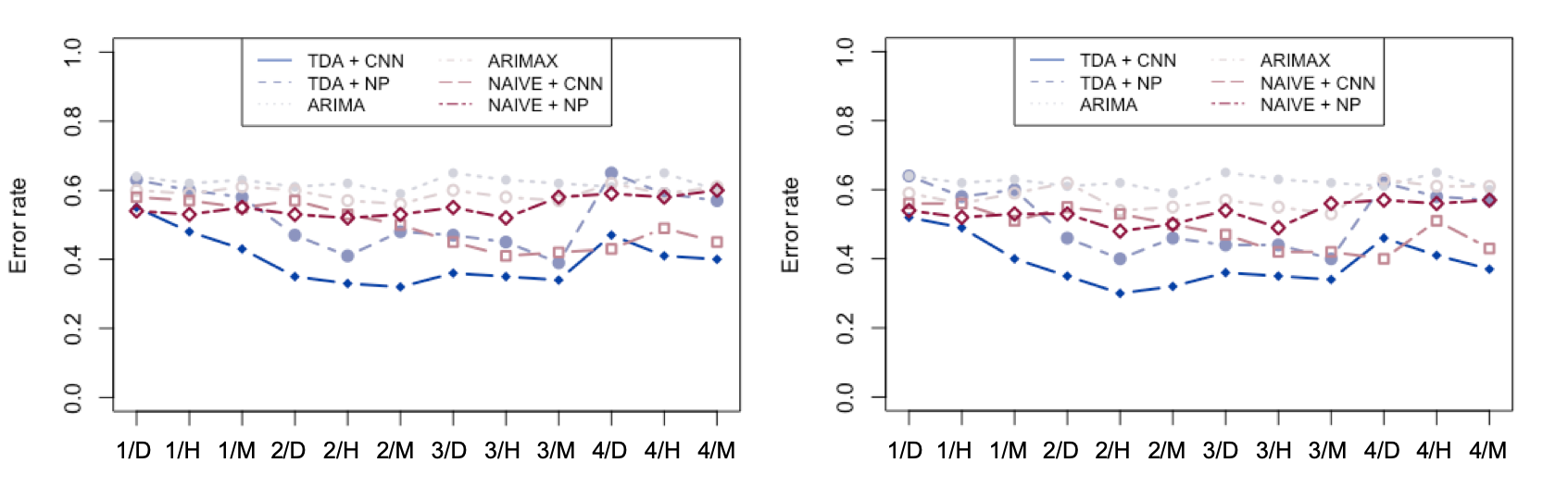}
		\captionof{figure}{Error rates without (\textit{left}) and with (\textit{right}) the covariate.}
		\label{fig:sim-bitcoin-pattern}
	\end{minipage}
\end{minipage}

\section{Conclusion}
In this study, we have proposed a novel time series featurization method by leveraging tools in TDA. 
To the best of our knowledge, our algorithm is the first systematic approach to successfully harness topological features of the attractor of the underlying dynamical system for arbitrary temporal data. Unlike traditional time series models, we do not impose any time-dependent structural assumptions . Importantly, we showed how the proposed method can effectively access the significant topological patterns while being robust to sampling noise by providing theoretical results.  
Based on the empirical results, we argue that our proposed method can be potentially very effective in identifying recurring patterns in granular time series data that are unknown but relevant to the target signal of interest.

\section{Acknowledgment}
We are grateful to Christopher Brookins, founder-CEO of Pugilist Ventures, for his assistance with the Bitcoin data acquisition and insightful comments.

\clearpage
\newpage

\bibliography{reference}

\clearpage
\newpage

\appendix

\section*{\centerline{APPENDIX}}

\section{Supplementary Materials for the Bitcoin Experiment}
\label{app:sim-results}

\subsection{Simulation Results}
\begin{table}[!ht] 
	\vspace{-.29cm}
	\renewcommand\thetable{1}
	\small 
	\begin{center}
		\begin{tabular}{l lll lll}
			\toprule
			\multirow{2}{*}{\quad \ \textbf{Method}} &\multicolumn{3}{l}{\quad Univariate} &\multicolumn{3}{l}{\ With Covariate} \\ 
			& \textbf{D} & \textbf{H} & \textbf{M}  & \textbf{D} & \textbf{H} & \textbf{M}\\
			\midrule
			TDA + CNN &  18.2 & 17.3 & 16.9 & 17.1 & 16.2 & 16.9 \\
			TDA + NP &  21.5 & 20.5 & 20.4 & 20.9 & 20.3 & 19.8 \\
			ARIMA &  28.5 & 33.2 & 36.1 &  - & -  &  - \\
			ARIMAX &  24.8 & 23.9 & 24.1 & 24.2 & 23.7 & 23.4 \\
			NAIVE + CNN & 20.3 & 19.7 & 18.6 & 19.5 & 20.8 &19.7 \\
			NAIVE + NP &  33.1 & 38.4 & 33.6 & 39.0 & 30.9 & 29.1 \\
			\bottomrule
	\end{tabular}\end{center}
	\normalsize
	\caption*{Table A1. Normalized RMSE ($\times 10^{-2}$) of the proposed and baseline methods}
	\label{result:prediction}
\end{table}

\begin{table}[!ht] 
	\vspace{-.245cm}
	\renewcommand\thetable{2}
	\small 
	\begin{center}
		\begin{tabular}{lcccccccccccc}
			\toprule
			\multirow{2}{*}{\quad \ \textbf{Method}} & \multicolumn{4}{c}{\textbf{Daily}} & \multicolumn{4}{c}{\textbf{Hourly}} & \multicolumn{4}{c}{\textbf{Every minute}}\tabularnewline
			\cmidrule{2-13} 
			& $P_{1}$ & $P_{2}$ & $P_{3}$ & $P_{4}$ & $P_{1}$ & $P_{2}$ & $P_{3}$ & $P_{4}$ & $P_{1}$ & $P_{2}$ & $P_{3}$ & $P_{4}$\tabularnewline
			\midrule
			\midrule
			\multirow{2}{*}{TDA + CNN} & 0.55&0.48&0.43&0.35&0.33&0.32&0.36&0.35&0.34&0.47&0.41&0.40 \tabularnewline
			& 0.52&0.49&0.40&0.35&0.30&0.32&0.36&0.35&0.34&0.46&0.41&0.37 \tabularnewline
			\cmidrule{2-13} 
			\multirow{2}{*}{TDA + NP} & 0.63&0.60&0.58&0.47&0.41&0.48&0.47&0.45&0.39&0.65&0.59&0.57 \tabularnewline
			& 0.64&0.58&0.60&0.46&0.40&0.46&0.44&0.44&0.40&0.62&0.58&0.57 \tabularnewline
			\cmidrule{2-13} 
			\multirow{2}{*}{ARIMA} & 0.64&0.62&0.63&0.61&0.62&0.59&0.65&0.63&0.62&0.61&0.65&0.60 \tabularnewline
			& -  & -  & -  & -  & -  & -  & -  & -  & -  & -  & -  & - \tabularnewline
			\cmidrule{2-13} 
			\multirow{2}{*}{ARIMAX} & 0.60&0.59&0.61&0.60&0.57&0.56&0.60&0.58&0.57&0.62&0.59&0.61 \tabularnewline
			& 0.59&0.56&0.59&0.62&0.54&0.55&0.57&0.55&0.53&0.63&0.61&0.61 \tabularnewline
			\cmidrule{2-13} 
			\multirow{2}{*}{NAIVE + CNN} & 0.58&0.57&0.55&0.57&0.53&0.50&0.45&0.41&0.42&0.43&0.49&0.45 \tabularnewline
			& 0.56&0.56&0.51&0.55&0.53&0.50&0.47&0.42&0.42&0.40&0.51&0.43 \tabularnewline
			\cmidrule{2-13} 
			\multirow{2}{*}{NAIVE + NP} & 0.54&0.53&0.55&0.53&0.52&0.53&0.55&0.52&0.58&0.59&0.58&0.6 \tabularnewline
			& 0.54&0.52&0.53&0.53&0.48&0.50&0.54&0.49&0.56&0.57&0.56&0.57 \tabularnewline
			\bottomrule
		\end{tabular}
	\end{center}
	\normalsize
	\caption*{Table A2. Classification error rate (the proportion of misclassified observations in test set) of the proposed and baseline methods. In each cell the lower number corresponds to the case that we incorporate additional covariate (trading volume).}
	\label{result:classification}
\end{table}

\subsection{Setup Details}

The analysis of cryptocurrency data is challenging. The embryonic nature of the technology, unreliable sentiment tracking, inadequate pricing and forecasting methodologies, and suspected price manipulation have created an incredibly complex and noisy dataset that is difficult to predict. It has very little systematic risk, which implies the cryptocurrency price itself should be the most valuable source of information about its price movement. Moreover, a study of \cite{segerestedt2018accuracy} has intimated that using more granular data would be more beneficial. These data characteristics offer a tremendously fertile ground for TDA to be effective in price prediction for cryptocurrency, rather than other traditional financial assets in which various asset pricing models are already available. As described in the main text, we use daily return of Bitcoin price  during the period from 01/01/2016 to 07/31/2018, at three different sampling rates.


\subsection{Setup}
In this section, we detail our simulation schemes and baseline methods for comparison. For the simulation, we set $y$ as a recurring pattern of our interest over single or multiple time periods. It can be a simple price change in the very next time point (e.g. whether price goes up/down) or more complex time series patterns over multiple time points (e.g. whether a rare price jump/drop occurs in the next 6 hours). We assume that there exists a set of time-series signals $\{f_1, f_2,...\}$ where each signal is highly likely to occur somehow prior to the predetermined pattern $y$ (see Figure \ref{fig:patterns-in-price}). Note that $y$ is known by user yet $f_j$'s are unknown. Moreover in reality, we only observe discrete samples from each $f_j$ with a potentially significant amount of noise. Our purpose is to effectively identify and featurize those $f_j$'s, $j=1,2,...,$ using Algorithm \ref{algorithm:TDA-featurization}, where we can utilize those features to build a model to predict a likelihood of occurrence of the pattern $y$.

\begin{figure}[!h] 
	\centering
	\includegraphics[width=0.8\columnwidth]{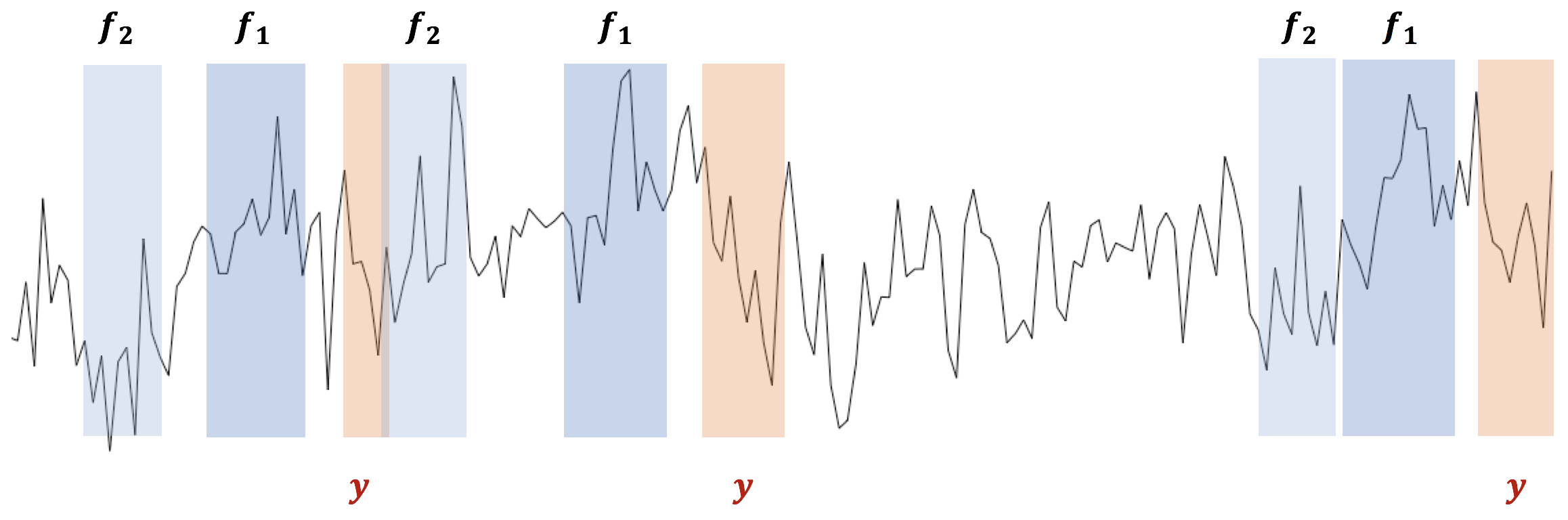}
	\caption{An example of recurring pattern $y$ and time-series signals $f_1, f_2$ that are very similar in shape and highly likely to occur prior to $y$. In reality, we only observe discrete samples from each $f_1, f_2$ with potentially significant amount of noise.  Data are obtained from minute-by-minute observations in the first 3 hours on 07/01/2018.}
	\label{fig:patterns-in-price}
\end{figure}

We let $N_{training}$ and $N_{test}$ denote the number of samples in training and test period respectively. To apply Algorithm \ref{algorithm:TDA-featurization}, we use a sequence of length $N \ll N_{training}$ to generate TDA features within the training period. Then we fit a machine learning model that predicts the predetermined pattern $y$ from the TDA features using samples in the training period, and use the trained model to get predicted $\hat{y}$ from samples in the test period. Then we compute a test error rate and roll over into next training, test window (rolled over by $N_{test}$). Throughout the simulation, we use $N_{training}=336$, $N_{test}=24$, and $N=168$.

As described in the main text, we use four patterns (P1)-(P4), each of which is either 3-class or binary classification problem. Figure \ref{fig:pattern-def-examples} illustrates an example of each pattern. For simulation, we use $k=6$ and $\alpha=0.1$. When we define a regime in ${P1}$, ${P4}$ we use thresholds to have almost the same number of samples in each class.

\begin{figure}[!ht] 
	\centering
	\includegraphics[width=0.8\columnwidth]{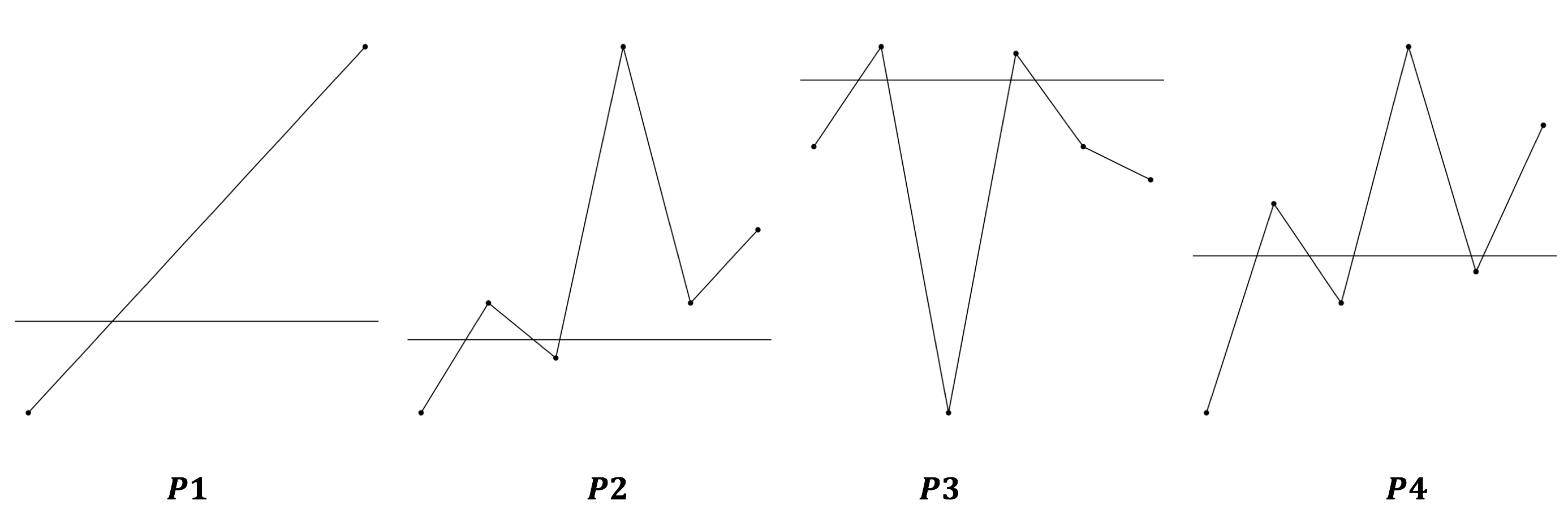}
	\caption{Examples of how the four patterns have their values. In the figure, we classify both $\textbf{P1}$ and $\textbf{P4}$ as \textit{up}, and a rare jump/drop \textit{occurs} in $\textbf{P2}$, $\textbf{P3}$ respectively.}
	\label{fig:pattern-def-examples}
\end{figure}

We construct TDA features via Algorithm \ref{algorithm:TDA-featurization}, where we try various combinations of $(m, \tau, d_{max})$ and report the best rate while we restrict the number of principal components to three ($l=3$). The average proportion of variance explained is roughly 40\%~50\%. For NP model, we use cross-validation-based superleaner ensemble algorithm \cite{van2007super} via the \texttt{SuperLearner} package in R to combine support vector machine, random forest, k-nearest neighbor regression, and multivariate adaptive regression splines. We use default methods (mostly cross-validation) in the package to tune any hyperparameters in the each of the models.


ARIMA and ARIMAX are two widely used time series models, where the latter incorporates an external regressor into ARIMA model. When there is no additional covariate we use the first $N$ coefficients of Fourier transform as the external regressor for ARIMAX. 

\section{A Brief Review of Dynamical Systems and Takens' Embedding Theorem}
\label{app:dynamic-systems}

Dynamical systems are mathematical objects used to model phenomena with states that evolve over time, and widely used in various scientific domains including quantum physics. We describe such dynamical evolutions as a set of state transition rules typically represented by differential equations. A time series can be considered a series of projections of the observed states from such a dynamical system, and therefore reconstructing the equations for the transition rules from the observed time series data is crucial to understand the underlying phenomena \cite{packard1980geometry}. The problem is that it is in general very difficult to fully reconstruct the equations for the transition rules of the dynamical systems from observed time series data without any a priori knowledge \cite{tong1990non, eckmann1992fundamental}. The most widely used approach to bypass this problem is to utilize an attractor. Putting it simply, an attractor is a set of numerical values toward which the given dynamical system eventually evolves over time, passing through various transitions determined by the system, without depending too much on its starting conditions. It is very common to have two time series signals that have the same transition rule but completely different observed waveforms. In contrast, the attractors of the dynamical system constructed from these time series still closely resemble each other. Thus studying the attractors provide a means to model dynamical systems (see, for example, \cite{liu2010chaotic, henry2001nonlinear, basharat2009time, ozaki19852, tong1990non, kantz2004nonlinear} for more details).

Generating an attractor directly from observed data is in theory impossible. This is basically due to the fact that an attractor is comprised of infinitely many points whereas we have only finitely many observed time series data. One well-known alternative is to form the {quasi-attractor} based on Takens' delay embedding theorem \cite{takens1981detecting}. Let $\mathcal{M}_0$ denote the manifold corresponding to the original dynamical system generating the observed time series data. Takens' delay embedding theorem guarantees the existence of a smooth map $\Psi$ such that $\mathcal{M}_0 \rightarrow \mathcal{M}^\prime$, where $\mathcal{M}^\prime$ is an $m$-dimensional Euclidean space and $\mathcal{M}_0$ and $\mathcal{M}^\prime$ are topologically equivalent. \cite{sauer1991embedology} found that this is guaranteed provided that $m > d_0$, where $d_0$ is the box-counting
dimension of the attractor in $\mathcal{M}_0$. \cite{robinson2005topological} generalized Takens' theorem to infinite-dimensional systems that have finite-dimensional attractors.

Hence, once we obtain $\Psi$ we can instead study $\mathcal{M}^\prime$ to analyze the underlying mechanism of given time series. This result has been popularized since it gives the ability to extract crucial information about time series of interest in Euclidean space. This embedding theorem has been proved to be useful to study chaotic or noisy time series (see the listed studies in the main text). In particular, it can be used for state space reconstruction of the original dynamical system so that TDA is applicable to point cloud data (e.g., \cite{umeda2017time,emrani2014persistent, venkataraman2016persistent,truong2017exploration}).

\section{More Background for Topological Data Analysis}
\label{app:tda_details}

This section introduces more detailed background in algebraic topology for topological data analysis that is used in this paper to supplement Section \ref{sec:tda}.

\subsection{Distance between sets on metric spaces}

When topological information of the underlying space is approximated
by the observed points, it is often needed to compare two sets with
respect to their metric structures. Here we present two distances
on metric spaces, Hausdorff distance and Gromov-Hausdorff distance.
We refer to \cite{BuragoBI2001} for more details and other distances.

The \emph{Hausdorff distance} is on sets embedded on the same metric
spaces. This distance measures how two sets are close to each other
in the embedded metric space. When $S\subset\mathfrak{X}$, we denote
by $U_{r}(S)$ the $r$-neighborhood of a set $S$ in a metric space,
i.e. $U_{r}(S)=\bigcup_{x\in S}\mathbb{B}_{\mathfrak{X}}(x,r)$.

\begin{definition}[Hausdorff distance] (\cite[Definition 7.3.1]{BuragoBI2001})
	
	\label{def:distance_hausdorff}
	
	Let $\mathfrak{X}$ be a metric space, and $X,Y\subset\mathfrak{X}$
	be a subset. The Hausdorff distance between $X$ and $Y$, denoted
	by $d_{H}(X,Y)$, is defined as 
	\[
	d_{H}(X,Y)=\inf\{r>0:\,X\subset U_{r}(Y)\text{ and }Y\subset U_{r}(X)\}.
	\]
	
\end{definition}

The \emph{Gromov-Hausdorff distance} measures how two sets are far
from being isometric to each other. To define the distance, we first
define a relation between two sets called \emph{correspondence}.

\begin{definition}
	
	Let $X$ and $Y$ be two sets. A \emph{correspondence} between $X$
	and $Y$ is a set $C\subset X\times Y$ whose projections to both
	$X$ and $Y$ are both surjective, i.e. for every $x\in X$, there
	exists $y\in Y$ such that $(x,y)\in C$, and for every $y\in Y$,
	there exists $x\in X$ with $(x,y)\in C$.
	
\end{definition}

For a correspondence, we define its \emph{distortion} by how the metric
structures of two sets differ by the correspondence.

\begin{definition}
	
	Let $X$ and $Y$ be two metric spaces, and $C$ be a correspondence
	between $X$ and $Y$. The \emph{distortion} of $C$ is defined by
	\[
	dis(C)=\sup\left\{ \left|d_{X}(x,x')-d_{Y}(y,y')\right|:\,(x,y),(x',y')\in C\right\} .
	\]
	
\end{definition}

Now the Gromov-Hausdorff distance is defined as the smallest possible
distortion between two sets.

\begin{definition}[Gromov-Hausdorff distance]
	(\cite[Theorem 7.3.25]{BuragoBI2001})
	
	\label{def:distance_gromov_hausdorff}
	
	Let $X$ and $Y$ be two metric spaces. The \emph{Gromov-Hausdorff
		distance} between $X$ and $Y$, denoted as $d_{GH}(X,Y)$, is defined
	as 
	\[
	d_{GH}(X,Y)=\frac{1}{2}\inf_{C}dis(C),
	\]
	where the infimum is over all correspondences between $X$ and $Y$.
	
\end{definition}

\subsection{Simplicial complex and Nerve Theorem}
\label{sec:simplicial_complex}
A simplicial complex can be seen as a high dimensional generalization
of a graph. Given a set $V$, an \textit{(abstract) simplicial complex}
is a set $K$ of finite subsets of $V$ such that $\alpha\in K$ and
$\beta\subset\alpha$ implies $\beta\in K$. Each set $\alpha\in K$
is called its \textit{simplex}. The \textit{dimension} of a simplex
$\alpha$ is $\dim\alpha=\mathrm{card}\alpha-1$, and the dimension
of the simplicial complex is the maximum dimension of any of its simplices.
Note that a simplicial complex of dimension $1$ is a graph.

When approximating the topology of the underlying space by observed samples, a common choice other than the Rips complex is the \textit{\v{C}ech complex,} defined next, Below, for any $x\in\mathbb{X}$ and $r>0$, we let $\mathbb{B}_{\mathbb{X}}(x,r)$
denote the closed ball centered at $x$ and radius $r>0$. 

\begin{definition}[\v{C}ech complex] \label{def:background_cech} Let
	$\mathcal{X}\subset\mathbb{X}$ be finite and $r>0$.
	The (weighted) \v{C}ech complex is the simplicial complex 
	\begin{equation}
	\textrm{\v{C}ech}^{\mathbb{X}}_{\mathcal{X}}(r):=\{ \sigma\subset\mathcal{X}:\ \cap_{x\in\sigma}\mathbb{B}_{\mathbb{X}}(x,r)\neq\emptyset\} ,\label{eq:background_cech}
	\end{equation}
	The superscript $\mathbb{X}$ will be dropped when understood from the
	context. \end{definition}

Note that the \v{C}ech complex and Rips complex have following interleaving
inclusion relationship 
\begin{equation}
\textrm{\v{C}ech}_{\mathcal{X}_{n}}(r)\subset R_{\mathcal{X}_{n}}(r)\subset\textrm{\v{C}ech}_{\mathcal{X}_{n}}(2r).\label{eq:background_interleaving_cechrips_general}
\end{equation}
In particular, when $\mathbb{X}$
is a Euclidean space, then the constant $2$ can be tightened to $\sqrt{2}$:
\begin{equation}
\textrm{\v{C}ech}_{\mathcal{X}_{n}}(r)\subset R_{\mathcal{X}_{n}}(r)\subset\textrm{\v{C}ech}_{\mathcal{X}_{n}}(\sqrt{2}r).\label{eq:background_interleaving_cechrips_euclidean}
\end{equation}

The topology of the \v{C}ech complex is linked to underlying continuous
spaces via Nerve Theorem. Let $r>0$
and consider the union of balls 
$
\cup_{x\in\mathcal{X}}\mathbb{B}_{\mathbb{X}}(x,r).
$
Then the union of balls
is homotopic equivalent to the \v{C}ech complex by the following Nerve Theorem.

\begin{theorem}[Nerve Theorem] \label{thm:background_nerve}
	
	Let $\mathcal{X}\subset\mathbb{X}$ be a finite set and $r>0$.
	Suppose for any finite subset $\{x_{1},\ldots,x_{k}\}\subset\mathcal{X}$, the
	intersection $\bigcap _{j=1}^{k}\mathbb{B}_{\mathbb{X}}(x_{j},r)$
	is either empty or contractible, then the \v{C}ech complex $\textrm{\v{C}ech}^{\mathbb{X}}_{\mathcal{X}}(r)$
	is homotopic equivalent to the union of balls $\cup_{x\in\mathcal{X}}\mathbb{B}_{\mathbb{X}}(x,r)$.
	
\end{theorem}

\subsection{Stability theorems}

Stability theorems have been established for various cases.

First, we consider the case when  the filtration $\mathcal{F}$
is generated from the sub-level sets or the super-level sets of a function.
Let $f,g:\mathbb{X}\to\mathbb{R}$ be two functions,
and let $PH_{*}(f)$ and $PH_{*}(g)$ be the corresponding persistent
homologies of the sublevel set filtrations $\{f\leq L\}_{L\in\mathbb{R}}$
and $\{g\leq L\}_{L\in\mathbb{R}}$.

We will impose a standard regularity condition for the functions $f$ and $g$, which is \emph{tameness}.

\begin{definition}[tameness] \label{def:stability_tame} (\cite[Section 3.8]{ChazalSGO2016})
	Let $f:\mathbb{X}\to\mathbb{R}$. Then $f$ is \emph{tame} if the image $im(\imath_{L}^{L'})$ of the homomorphism $\imath_{L}^{L'}:H_{k}(f^{-1}(-\infty,L])\to H_{k}(f^{-1}(-\infty,L'])$ induced from the inclusion $\imath_{L}^{L'}:f^{-1}(-\infty,L] \to f^{-1}(-\infty,L']$ is of finite rank for all $k\in\mathbb{N}\cup\{0\}$ and $L<L'$.
\end{definition}

When two functions $f$ and $g$ satisfy the tameness condition, their bottleneck distance is bounded by their $\ell_{\infty}$ distance, an important and useful fact known as the stability theorem.

\begin{theorem}[Stability theorem for sublevel or superlevel sets of a function] \cite{CohenEH2007, ChazalCGGO2009, ChazalSGO2016}
	\label{thm:stability_function} For two tame functions
	$f,g:\mathbb{X}\to\mathbb{R}$, 
	\[
	d_{B}(PH_{k}(f),PH_{k}(g))\leq\|f-g\|_{\infty}.
	\]
\end{theorem}

Second, we consider the case for Rips filtration. Let $X$ and $Y$ be two metric spaces, and let $PH_{k}(R_{X})$ and $PH_{k}(R_{Y})$ be the corresponding persistent homologies of the Rips complex filtrations $\{R_{X}(r)\}_{r\in\mathbb{R}}$ and $\{R_{Y}(r)\}_{r\in\mathbb{R}}$.

We say that the metric space is totally bounded if it can be arbitrarily approximated by a finite set of points. For example, a bounded subset of Euclidean space is totally bounded.

\begin{definition}
	A metric space $X$ is \emph{totally bounded} if for any $\epsilon>0$, there exists a finite set of points $x_{1},\ldots,x_{n}\in X$ that $\epsilon$-approximates $X$, i.e. for all $x\in X$, there exists $x_{i}$ such that $d(x,x_{i})<\epsilon$.
\end{definition}

Now, when two metric spaces $X$ and $Y$ are totally bounded, then their bottleneck distance between corresponding Rips persistence homologies is bounded by their Gromov-Hausdorff distance, as the following stability theorem.

\begin{theorem}[Stability theorem for Rips complex] \cite[Theorem 5.2]{ChazalSO2012}
	\label{thm:stability_rips} Let $X$ and $Y$ be two totally bounded metric spaces. Then, 
	\[
	d_{B}(PH_{k}(R_{X}),PH_{k}(R_{Y}))\leq d_{GH}(X,Y).
	\]
\end{theorem}

\section{Two examples in which PCA reduces homological noises}
\label{sec:appendix-PCA-example}

This section presents two examples for Section \ref{sec:pca-denoise} in which PCA reduces homological noises. Suppose the PCA uses the first $l$ principal components.

Figure \ref{fig:pca_homology_high} shows how the topological noise of higher dimension is reduced. The topological noise is topologically sphere of dimension $l$ or higher. When projected to a $l$-dimensional linear space, this sphere reduces to a disk of dimension $l$.  Hence, this topological noise is eliminated after the PCA.

\begin{figure} [!h]
	\begin{center}
		\begin{subfigure}{0.40\linewidth}
			\includegraphics[width=\linewidth]{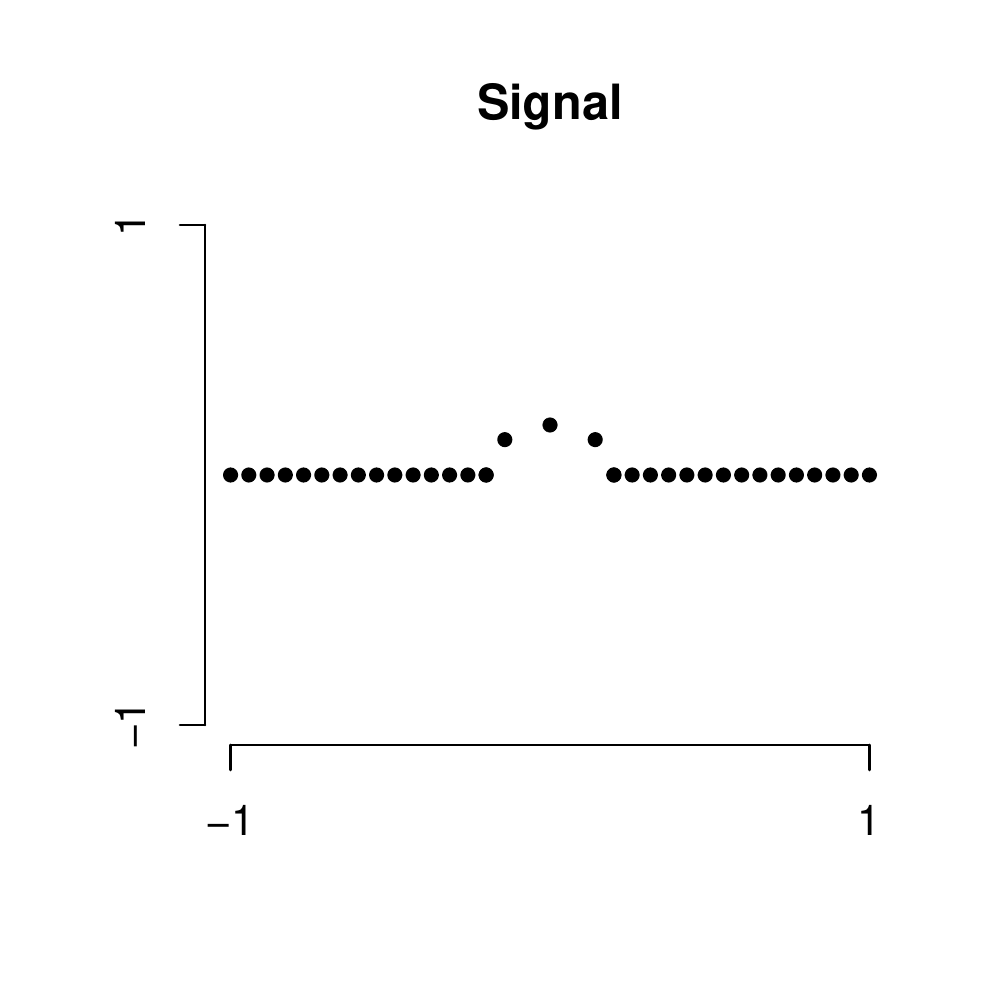}
		\end{subfigure}
		\begin{subfigure}{0.40\linewidth}
			\includegraphics[width=\linewidth]{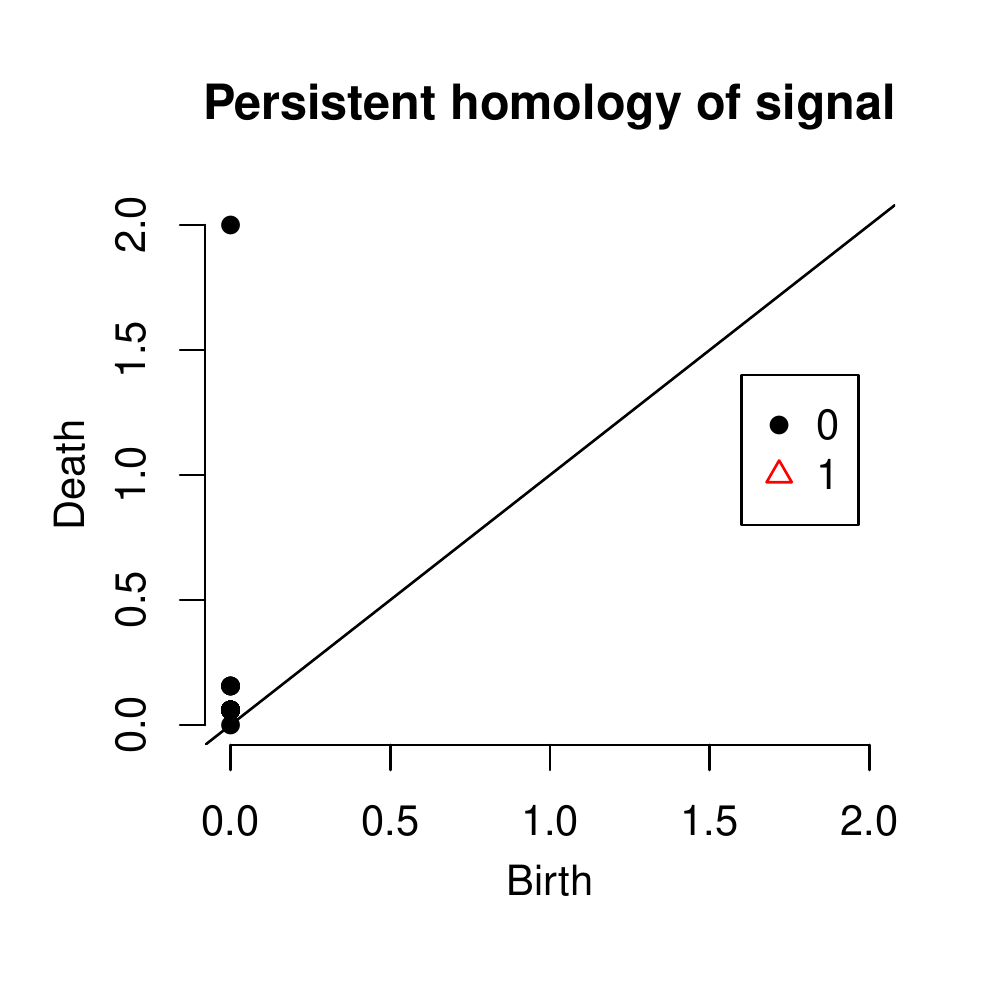}
		\end{subfigure}
		\begin{subfigure}{0.40\linewidth}
			\includegraphics[width=\linewidth]{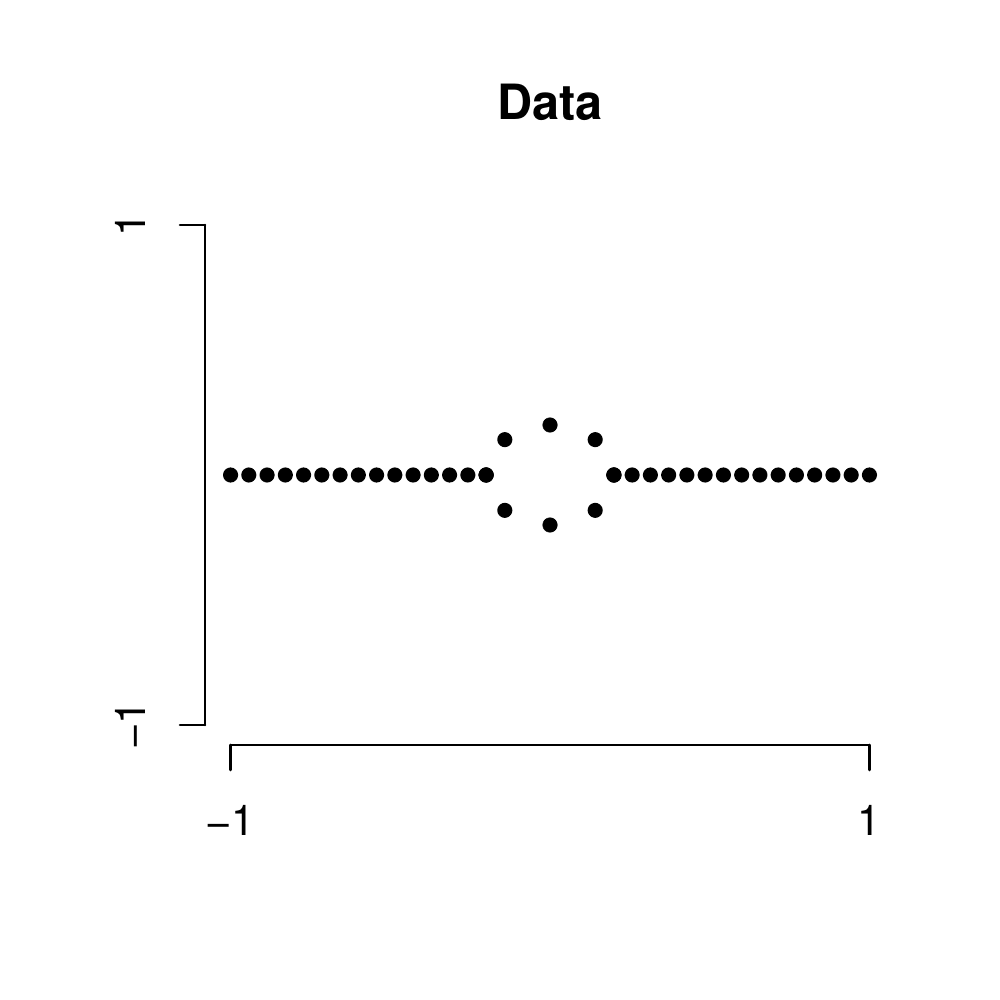}
		\end{subfigure}
		\begin{subfigure}{0.40\linewidth}
			\includegraphics[width=\linewidth]{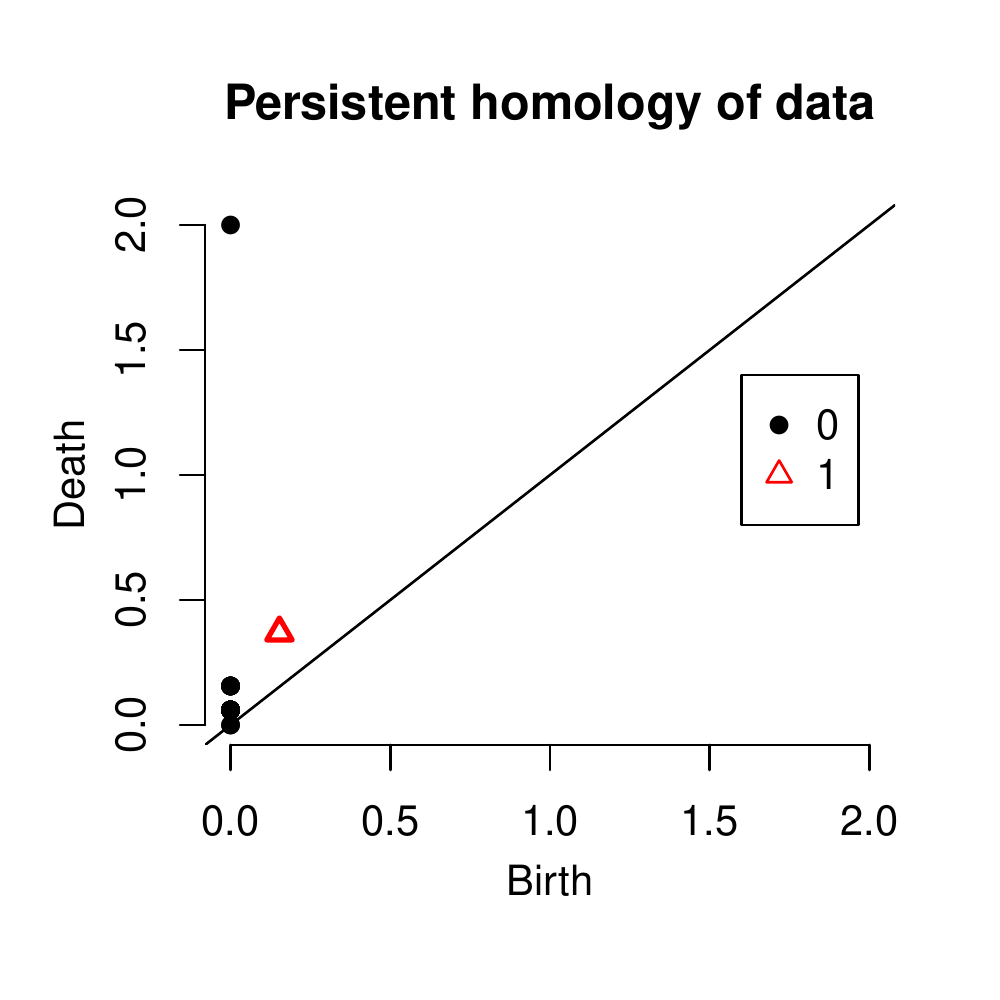}
		\end{subfigure}
		\begin{subfigure}{0.40\linewidth}
			\includegraphics[width=\linewidth]{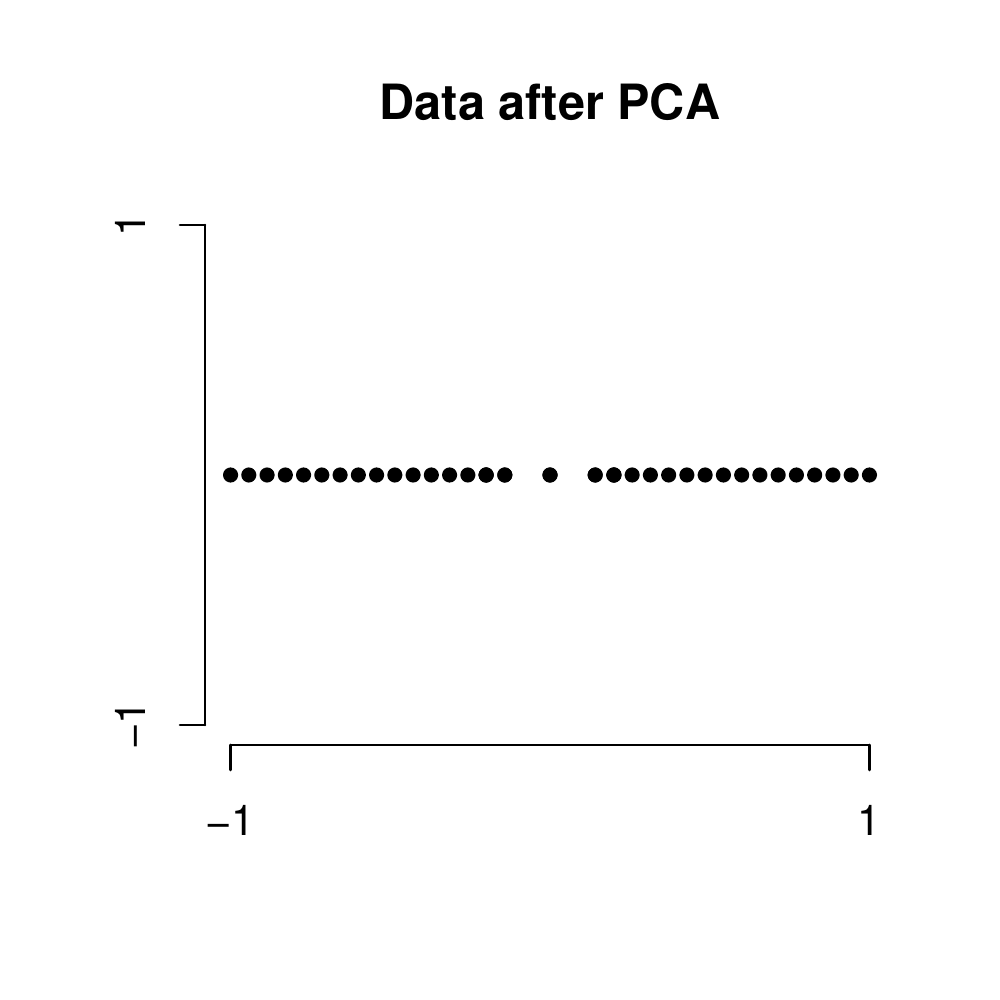}
		\end{subfigure}
		\begin{subfigure}{0.40\linewidth}
			\includegraphics[width=\linewidth]{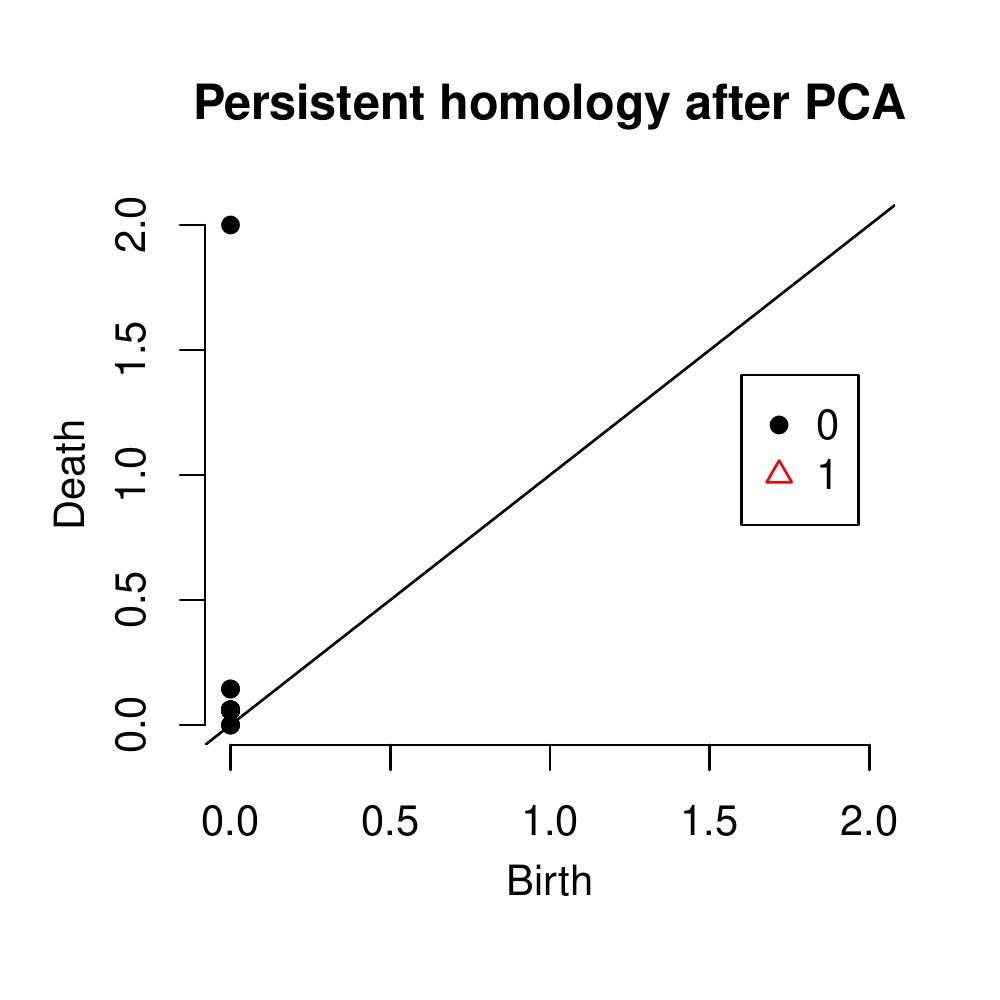}
		\end{subfigure}
	\end{center}
	\vspace{-.25 in}
	\caption{An example in which PCA can reduce the higher dimensional homological noise while preserving the homological features from the signal. We focus on the $0$-dimensional features, marked as black points on the right column. The signal consists of 35 points on a union of two intervals and a half circle (top left). This has only a $0$-dimensional feature in the persistence diagram (top right). Noise points are added so that a circle is formed in the center (mid left). Then a noisy $1$-dimensional feature appears near the diagonal line in the persistence diagram (mid right). After doing PCA to project to $1$-dimensional space (bottom left), all the noisy $1$-dimensional features are gone but the signal $0$-dimensional features are preserved in the persistence diagram (bottom right).} 
	\label{fig:pca_homology_high}
\end{figure}

On the other hand, Figure \ref{fig:pca_homology} shows how the topological noise of lower dimension is reduced. The topological noise is of dimension smaller than $l$ aligned with the orthogonal direction of the linear subspace. Then, the topological noise along the orthogonal direction is reduced, and due to its alignment, the topological noise is eliminated entirely.

\begin{figure} [!h]
	\begin{center}
		\begin{subfigure}{0.40\linewidth}
			\includegraphics[width=\linewidth]{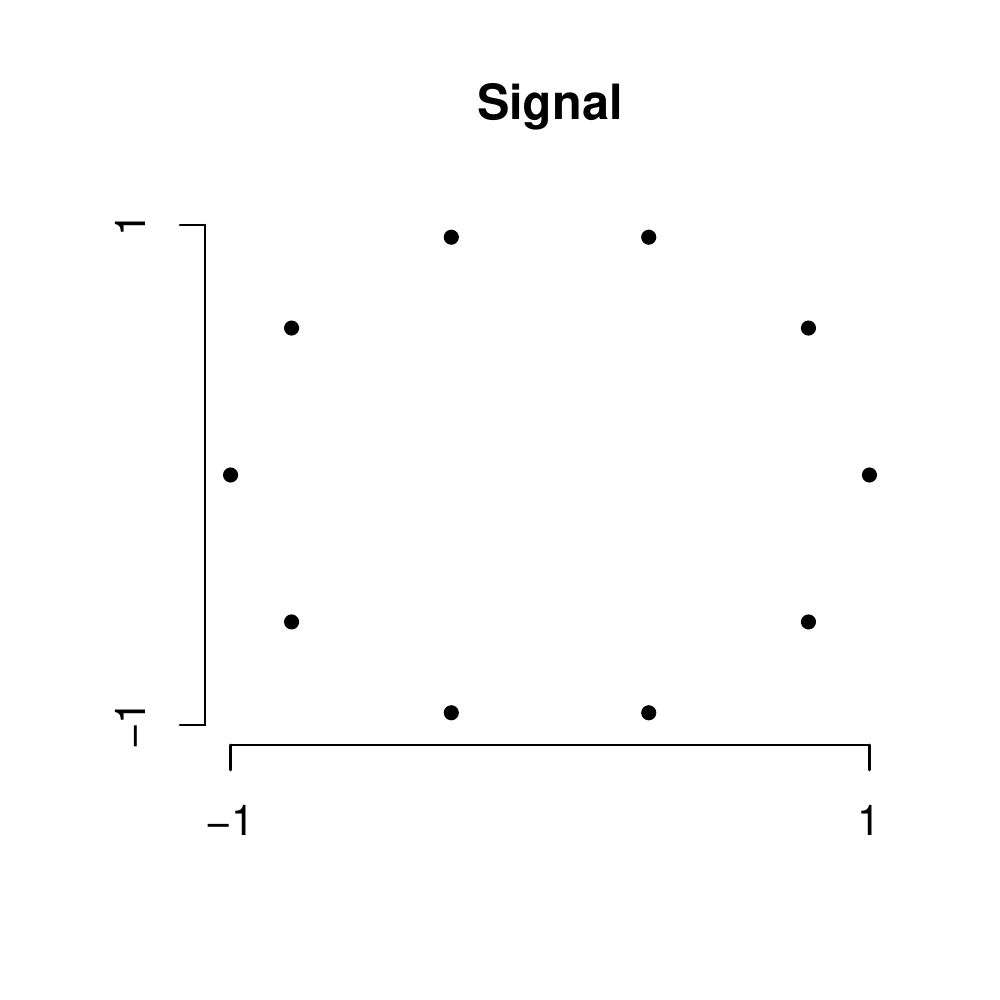}
		\end{subfigure}
		\begin{subfigure}{0.40\linewidth}
			\includegraphics[width=\linewidth]{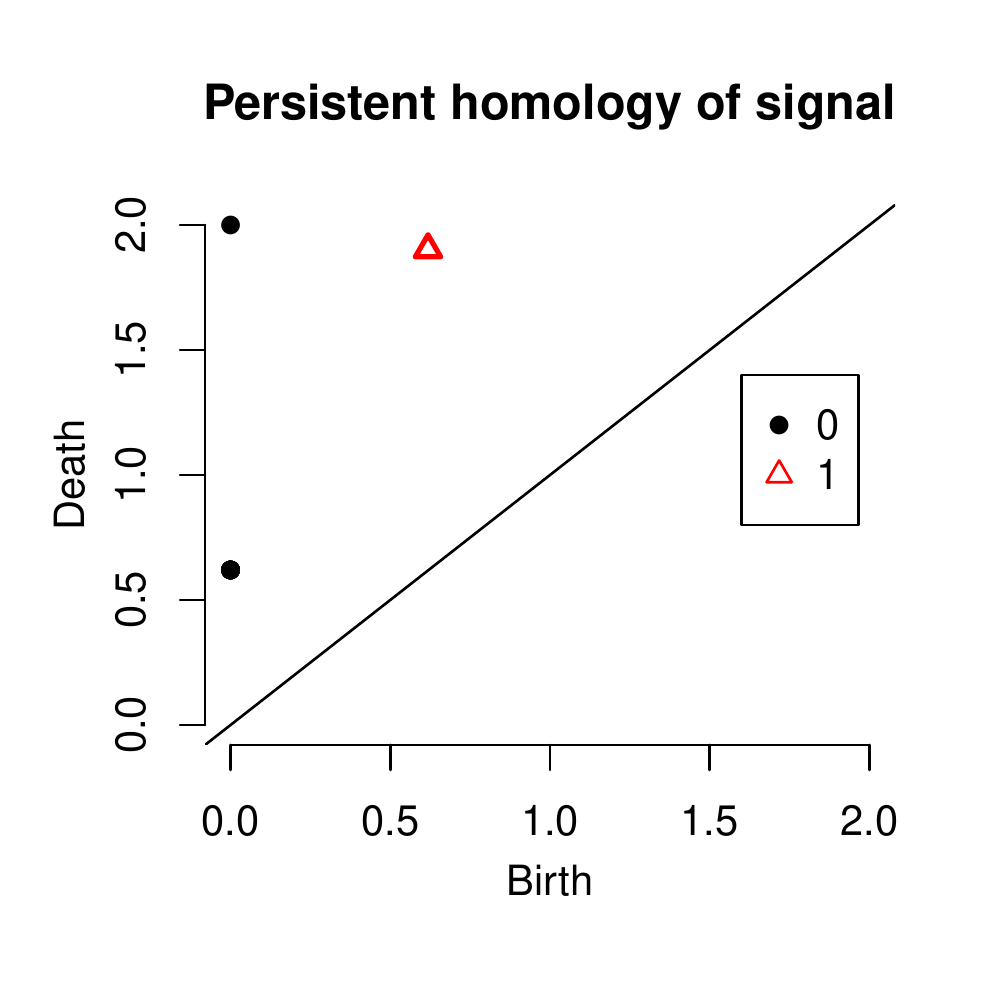}
		\end{subfigure}
		\begin{subfigure}{0.40\linewidth}
			\includegraphics[width=\linewidth]{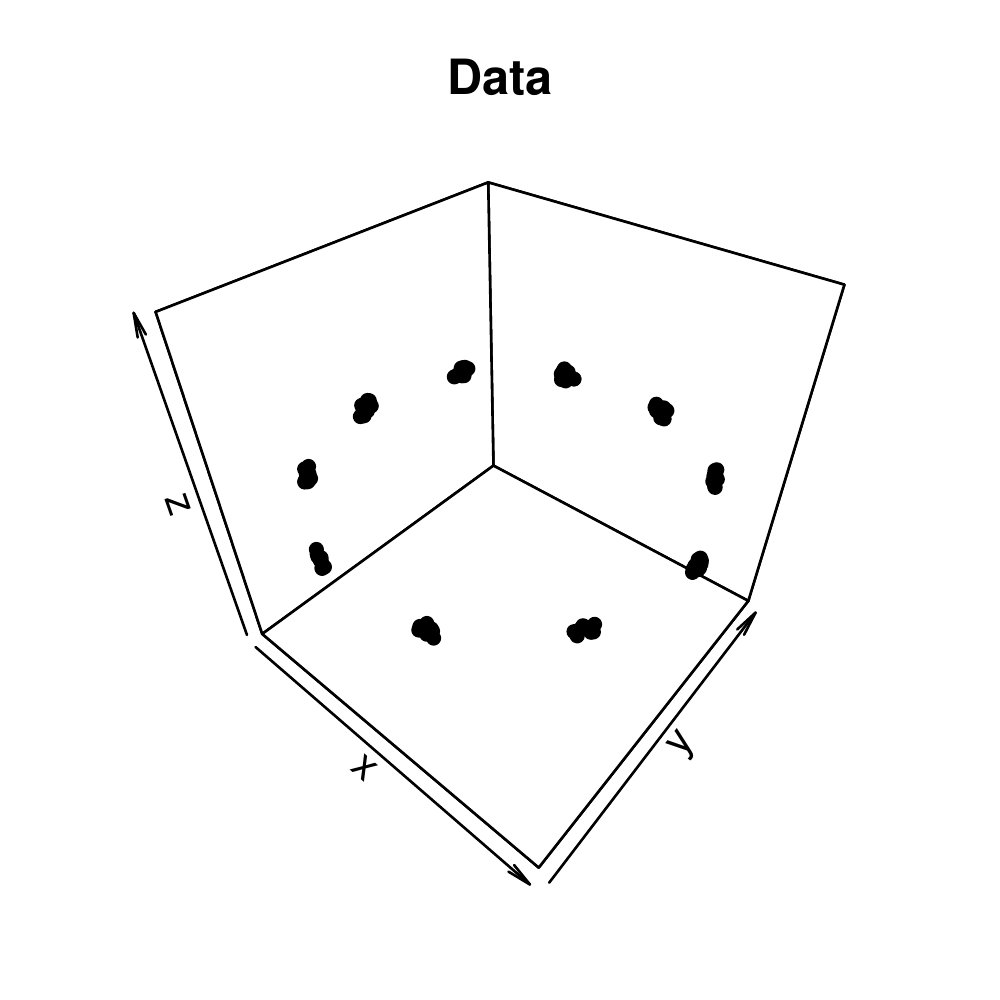}
		\end{subfigure}
		\begin{subfigure}{0.40\linewidth}
			\includegraphics[width=\linewidth]{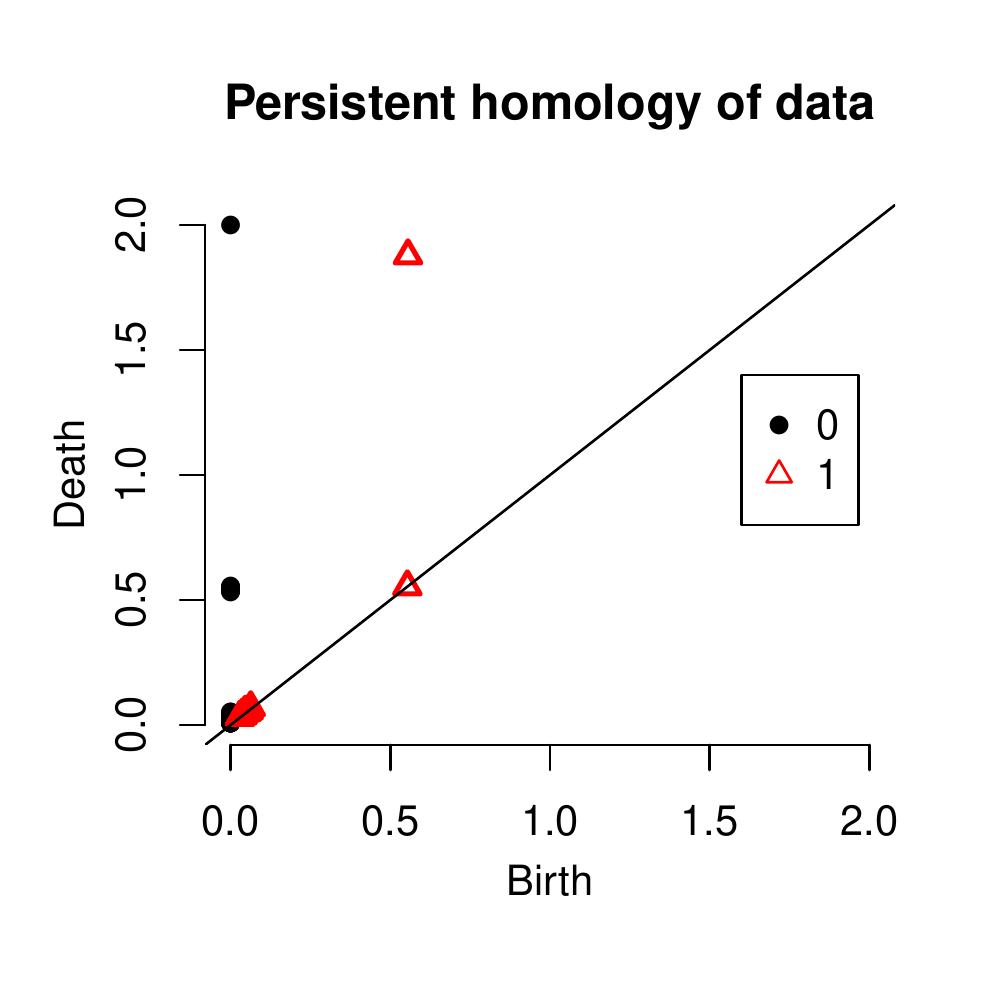}
		\end{subfigure}
		\begin{subfigure}{0.40\linewidth}
			\includegraphics[width=\linewidth]{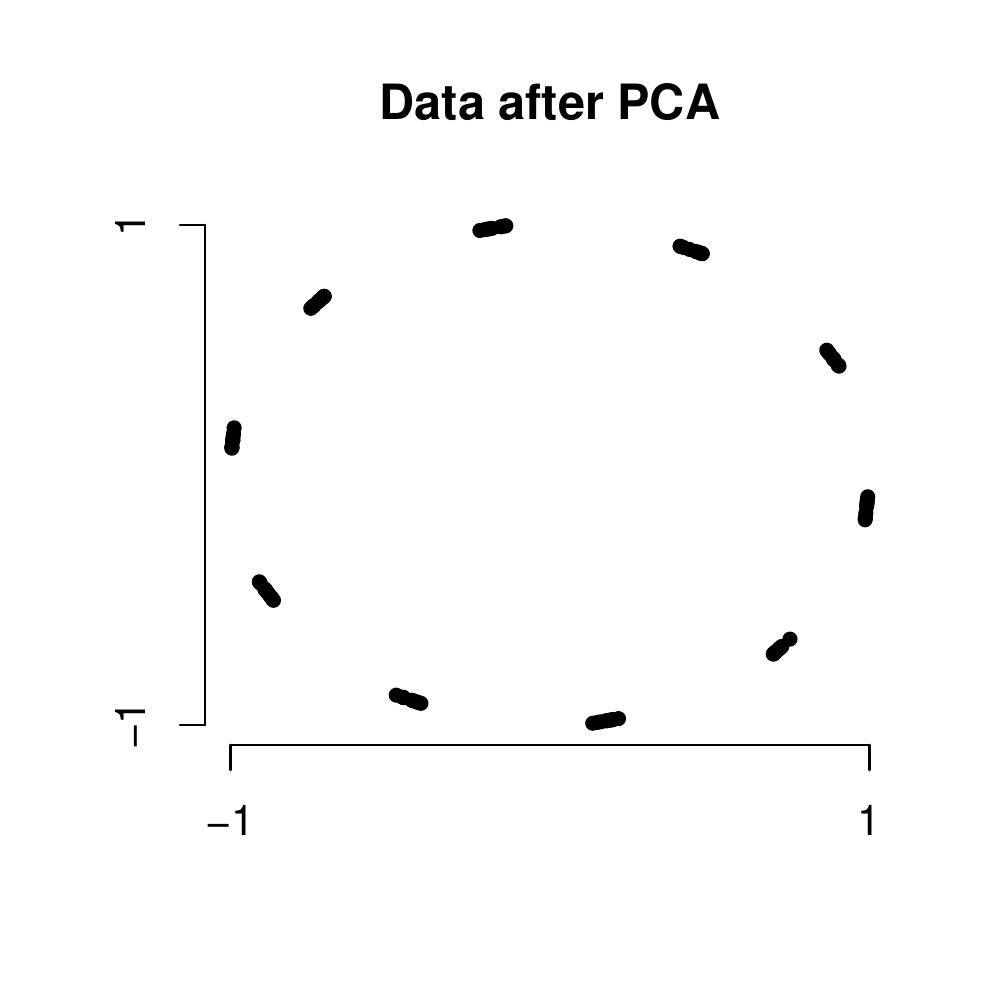}
		\end{subfigure}
		\begin{subfigure}{0.40\linewidth}
			\includegraphics[width=\linewidth]{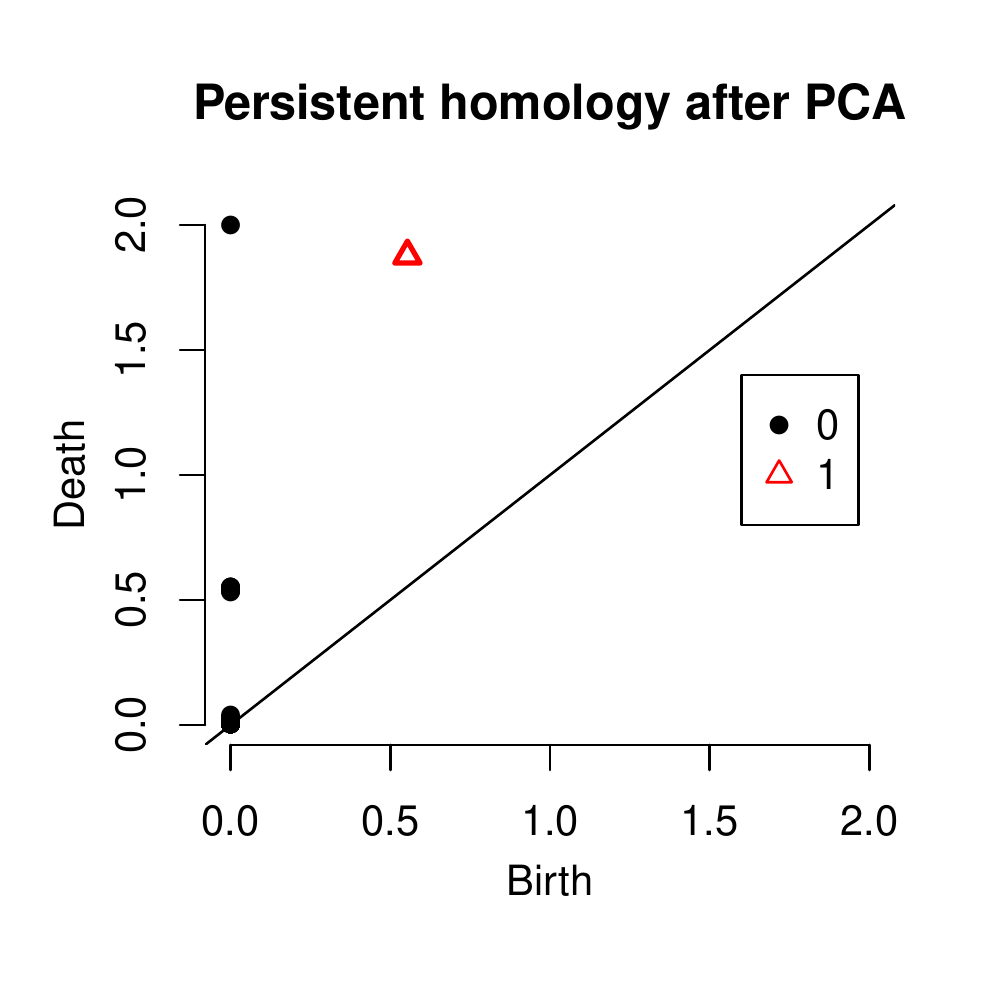}
		\end{subfigure}
	\end{center}
	\vspace{-.25 in}
	\caption{An example in which PCA can reduce the lower dimensional homological noise -- assumed for convenience to be aligned orthogonally to the true signal space -- while preserving the homological features from the signal. We focus on the $1$-dimensional feature, marked as red triangles on the right column. The signal consists of 10 equally spaced points on the unit circle (top left). This has one feature in the persistence diagram (top right). Noise is added to the tangential direction to the circle and $z$-axis to sample $100$ points (mid left). Then several noisy features appear near the diagonal line in the persistence diagram (mid right). After doing PCA to project to $2$-dimensional space (bottom left), all the noisy features are gone but the signal feature is preserved in the persistence diagram (bottom right).} 
	\label{fig:pca_homology}
\end{figure}

\clearpage

\section{Proofs of Section \ref{sec:pca-denoise}}

\begin{proof}[Proof of Proposition \ref{prop:pca_stability_persistent_homology_linear}]
	
	Note that applying PCA is isometric to $X$ being projected to the
	linear subspace generated by the first $l$ eigenvectors of $X^{\top}X$.
	For $k\leq d$, let $\mathbb{V}^{k}$ and $V^{k}$ be the linear subspaces
	generated by the first $k$ eigenvectors of $\mathbb{X}^{\top}\mathbb{X}$
	and $X^{\top}X$, respectively, and let $\Pi_{\mathbb{V}^{k}}$ and
	$\Pi_{V^{k}}$ be the projection operator to $\mathbb{V}^{k}$ and
	$V^{k}$, respectively. With this notation, $X^{l}$ is isometric
	to $\Pi_{V^{l}}(X)$, where $X$ is understood as a point cloud. And
	hence $d_{GH}(X^{l},\mathbb{X})$ can be expanded as 
	\[
	d_{GH}(X^{l},\mathbb{X})=d_{GH}(\Pi_{V^{l}}(X),\mathbb{X}).
	\]
	Note that from \cite[Theorem 2]{YuWS2015}, 
	\begin{align}
	\left\Vert (Id-\Pi_{V^{\tilde{l}}})\Pi_{\mathbb{V}^{\tilde{l}}}\right\Vert _{F} & \leq\frac{2\left\Vert \frac{1}{n}X^{\top}X-\frac{1}{n}\mathbb{X}^{\top}\mathbb{X}\right\Vert _{F}}{\lambda_{\tilde{l}}}.\nonumber \\
	& \leq\frac{2\sum_{i=1}^{n}\left\Vert X_{i}X_{i}^{\top}-\mathbb{X}_{i}\mathbb{X}_{i}^{\top}\right\Vert _{F}}{n\lambda_{\tilde{l}}}.\label{eq:analysis_stability_david_kahan-1-1}
	\end{align}
	Now, let $\epsilon_{i}:=X_{i}-\mathbb{X}_{i}$ for each $i=1,\ldots,n$.
	Then, 
	\begin{align*}
	\left\Vert X_{i}X_{i}^{\top}-\mathbb{X}_{i}\mathbb{X}_{i}^{\top}\right\Vert _{F} & =\left\Vert (\mathbb{X}_{i}+\epsilon_{i})(\mathbb{X}_{i}+\epsilon_{i})^{\top}-\mathbb{X}_{i}\mathbb{X}_{i}^{\top}\right\Vert _{F}\\
	& =\left\Vert \mathbb{X}_{i}\epsilon_{i}^{\top}\right\Vert _{F}+\left\Vert \epsilon_{i}\mathbb{X}_{i}^{\top}\right\Vert _{F}+\left\Vert \epsilon_{i}\epsilon_{i}^{\top}\right\Vert _{F}\\
	& \leq2\left\Vert \mathbb{X}_{i}\right\Vert _{2}\left\Vert \epsilon_{i}\right\Vert _{2}+\left\Vert \epsilon_{i}\right\Vert _{2}^{2}.
	\end{align*}
	Hence applying to \eqref{eq:analysis_stability_david_kahan-1-1} gives
	the bound as 
	\begin{align*}
	\left\Vert (Id-\Pi_{V^{\tilde{l}}})\Pi_{\mathbb{V}^{\tilde{l}}}\right\Vert _{F} & \leq\frac{2\sup_{i}\left\Vert \epsilon_{i}\right\Vert _{2}\left(\sup_{i}\left\Vert \epsilon_{i}\right\Vert _{2}+2\sup_{i}\left\Vert \mathbb{X}_{i}\right\Vert _{2}\right)}{\lambda_{\tilde{l}}}.
	\end{align*}
	Now, note that the distance between $\Pi_{V^{l}}(X_{i})$ (after PCA)
	and $\mathbb{X}_{i}^{N}$ can be bounded as 
	\begin{align}
	\left\Vert \Pi_{V^{l}}(X_{i})-\mathbb{X}_{i}\right\Vert _{2} & =\left\Vert \Pi_{V^{l}}(\epsilon_{i})+\Pi_{V^{l}}(\mathbb{X}_{i})-\mathbb{X}_{i}\right\Vert _{2}\nonumber \\
	& \leq\left\Vert (Id-\Pi_{V^{l}})(\mathbb{X}_{i})\right\Vert _{2}+\left\Vert \Pi_{\mathbb{V}^{l}}(\epsilon_{i})\right\Vert _{2}.\label{eq:analysis_stability_pca_error_decomposition-1-1}
	\end{align}
	Then $V^{\tilde{l}}\subset V^{l}$ implies that $\left\Vert (Id-\Pi_{V^{l}})v\right\Vert _{2}\leq\left\Vert (Id-\Pi_{V^{\tilde{l}}})v\right\Vert _{2}$,
	and $\mathbb{X}_{i}\in\mathbb{V}^{\tilde{l}}$ implies
	$\Pi_{\mathbb{V}^{\tilde{l}}}(\mathbb{X}_{i})=\mathbb{X}_{i}$, and
	hence the first term of \eqref{eq:analysis_stability_pca_error_decomposition-1-1}
	is bounded as 
	\begin{align*}
	\left\Vert (Id-\Pi_{V^{l}})(\mathbb{X}_{i})\right\Vert _{2} & \leq\left\Vert (Id-\Pi_{V^{\tilde{l}}})(\mathbb{X}_{i})\right\Vert _{2}=\left\Vert (Id-\Pi_{V^{\tilde{l}}})\Pi_{\mathbb{V}^{\tilde{l}}}(\mathbb{X}_{i})\right\Vert _{2}\\
	& \leq\left\Vert (Id-\Pi_{V^{\tilde{l}}})\Pi_{\mathbb{V}^{\tilde{l}}}(\mathbb{X}_{i})\right\Vert _{2}\\
	& \leq\left\Vert (Id-\Pi_{V^{\tilde{l}}})\Pi_{\mathbb{V}^{\tilde{l}}}\right\Vert _{2}\left\Vert \mathbb{X}_{i}\right\Vert _{2}.
	\end{align*}
	Then further applying \eqref{eq:analysis_stability_david_kahan-1-1}
	gives 
	\begin{equation}
	\left\Vert (Id-\Pi_{V^{l}})(\mathbb{X}_{i})\right\Vert _{2}\leq\frac{2\sup_{i}\left\Vert \epsilon_{i}\right\Vert _{2}\sup_{i}\left\Vert \mathbb{X}_{i}\right\Vert _{2}\left(\sup_{i}\left\Vert \epsilon_{i}\right\Vert _{2}+2\sup_{i}\left\Vert \mathbb{X}_{i}\right\Vert _{2}\right)}{\lambda_{\tilde{l}}}.\label{eq:analysis_stability_pca-1-1}
	\end{equation}
	Now, the second term of \eqref{eq:analysis_stability_pca_error_decomposition-1-1}
	is bounded as 
	\begin{equation}
	\left\Vert \Pi_{\mathbb{V}^{l}}(\epsilon_{i})\right\Vert _{2}\leq\left\Vert \epsilon_{i}\right\Vert _{2}\leq\sup_{i}\left\Vert \epsilon_{i}\right\Vert _{2}.\label{eq:analysis_stability_error-1-1}
	\end{equation}
	Then by applying \eqref{eq:analysis_stability_pca-1-1} and \eqref{eq:analysis_stability_error-1-1}
	to \eqref{eq:analysis_stability_pca_error_decomposition-1-1}, the
	Hausdorff distance between $\Pi_{V^{l}}(X)$ and $\mathbb{X}$
	can be bounded as 
	\begin{align}
	d_{H}(\Pi_{V^{l}}(X),\mathbb{X}) & \leq\sup_{1\leq i\leq n}\left\Vert \Pi_{V^{l}}(X_{i})-\mathbb{X}_{i}\right\Vert _{2}\nonumber \\
	& \leq\sup_{i}\left\Vert \epsilon_{i}\right\Vert _{2}\left(1+\frac{2\sup_{i}\left\Vert \mathbb{X}_{i}\right\Vert _{2}\left(\sup_{i}\left\Vert \epsilon_{i}\right\Vert _{2}+2\sup_{i}\left\Vert \mathbb{X}_{i}\right\Vert _{2}\right)}{\lambda_{\tilde{l}}}\right).\label{eq:analysis_stability_hausdorff_pca-1-1}
	\end{align}
	And correspondingly,
	\begin{align*}
	d_{B}(D_{\mathbb{X}},D_{X^{l}}) & \leq d_{GH}(X^{l},\mathbb{X})=d_{GH}(\Pi_{V^{l}}(X),\mathbb{X})\\
	& \leq\sup_{i}\left\Vert \mathbb{X}_{i}-X_{i}\right\Vert _{2}\left(1+\frac{2\sup_{i}\left\Vert \mathbb{X}_{i}\right\Vert _{2}\left(\sup_{i}\left\Vert \mathbb{X}_{i}-X_{i}\right\Vert _{2}+2\sup_{i}\left\Vert \mathbb{X}_{i}\right\Vert _{2}\right)}{\lambda_{\tilde{l}}}\right).
	\end{align*}
	
\end{proof}

\begin{proof}[Proof of Corollary \ref{cor:pca_stability_persistent_homology_general}]
	
	Note that $d_{B}(D_{\mathbb{X}},D_{X^{l}})$ can be bounded as 
	\[
	d_{B}(D_{\mathbb{X}},D_{X^{l}})\leq d_{B}(D_{\mathbb{X}},D_{\mathbb{X}_{\mathbb{V}}})+d_{B}(D_{\mathbb{X}_{\mathbb{V}}},D_{X^{l}}).
	\]
	Then the first term is bounded as 
	\[
	d_{B}(D_{\mathbb{X}},D_{\mathbb{X}_{\mathbb{V}}})\leq d_{H}(\mathbb{X},\mathbb{X}_{\mathbb{V}})=d_{H}(\mathbb{X},\mathbb{V}).
	\]
	And by using $\left\Vert (\mathbb{X}_{\mathbb{V}})_{i}\right\Vert _{2}\leq\left\Vert \mathbb{X}_{i}\right\Vert _{2}$
	and $\sup_{i}\left\Vert (\mathbb{X}_{\mathbb{V}})_{i}-X_{i}\right\Vert _{2}\leq d_{H}(\mathbb{X},\mathbb{V})+\sup_{i}\left\Vert \mathbb{X}_{i}-X_{i}\right\Vert _{2}$,
	the second term is bounded as
	\begin{align*}
	& d_{B}(D_{\mathbb{X}_{\mathbb{V}}},D_{X^{l}})\\
	& \leq\sup_{i}\left\Vert (\mathbb{X}_{\mathbb{V}})_{i}-X_{i}\right\Vert _{2}\left(1+\frac{2\sup_{i}\left\Vert (\mathbb{X}_{\mathbb{V}})_{i}\right\Vert _{2}\left(\sup_{i}\left\Vert (\mathbb{X}_{\mathbb{V}})_{i}-X_{i}\right\Vert _{2}+2\sup_{i}\left\Vert (\mathbb{X}_{\mathbb{V}})_{i}\right\Vert _{2}\right)}{\lambda_{\tilde{l}}}\right)\\
	& \leq\left(d_{H}(\mathbb{X},\mathbb{V})+\sup_{i}\left\Vert \mathbb{X}_{i}-X_{i}\right\Vert _{2}\right)\\
	& \ \times\left(1+\frac{2\sup_{i}\left\Vert \mathbb{X}_{i}\right\Vert _{2}\left(d_{H}(\mathbb{X},\mathbb{V})+\sup_{i}\left\Vert \mathbb{X}_{i}-X_{i}\right\Vert _{2}+2\sup_{i}\left\Vert \mathbb{X}_{i}\right\Vert _{2}\right)}{\lambda_{\tilde{l}}}\right).
	\end{align*}
	Hence combining these two terms gives 
	\begin{align*}
	d_{B}(D_{\mathbb{X}},D_{X^{l}})
	& \leq d_{B}(D_{\mathbb{X}},D_{\mathbb{X}_{\mathbb{V}}})+d_{B}(D_{\mathbb{X}_{\mathbb{V}}},D_{X^{l}})\\
	& \leq d_{H}(\mathbb{X},\mathbb{V})+\left(d_{H}(\mathbb{X},\mathbb{V})+\sup_{i}\left\Vert \mathbb{X}_{i}-X_{i}\right\Vert _{2}\right)\\
	& \times\left(1+\frac{2\sup_{i}\left\Vert \mathbb{X}_{i}\right\Vert _{2}\left(d_{H}(\mathbb{X},\mathbb{V})+\sup_{i}\left\Vert \mathbb{X}_{i}-X_{i}\right\Vert _{2}+2\sup_{i}\left\Vert \mathbb{X}_{i}\right\Vert _{2}\right)}{\lambda_{\tilde{l}}}\right).
	\end{align*}
	
\end{proof}

\section{Proofs of Section \ref{sec:stability-thm}}
\label{app:proof-stability}

\begin{proof}[Proof of Proposition \ref{prop:analysis_stability_persistent_homology}] 
	
	When $\left\Vert \epsilon\right\Vert _{\infty}=\infty$, then there
	is nothing to prove, so we can assume that $\left\Vert \epsilon\right\Vert _{\infty}<\infty$.
	Then since $f$ is a Lipschitz function on a bounded domain $[0,T]$,
	$\mathbb{X}$ is bounded as well. 
	
	
	Let $\mathbb{X}^{N}:=\{SW_{m,\tau}f(t):\,t=(0+(m-1)\tau),(1+(m-1)\tau),\ldots,T\}$
	and write as a matrix where each element corresponds to each row.
	Denote $i$th rows of $X$ and $\mathbb{X}^{N}$ as $X_{i}$ and $\mathbb{X}_{i}^{N}$.
	Let $\tilde{N}:=N-(m-1)\tau$, and understand $X$ and $\mathbb{X}^{N}$
	as a point cloud of $\{X_{i}\}_{i\leq\tilde{N}}$ and $\{\mathbb{X}_{i}^{N}\}_{i\leq\tilde{N}}$,
	respectively, depending on the context.
	
	We first consider bounding the Gromov-Hausdorff distance $d_{GH}(X^{l},\mathbb{X})$
	between the data after PCA $X^{l}$ and the Takens embedding of the
	signal $\mathbb{X}$. Note first that $d_{GH}(X^{l},\mathbb{X})$
	is bounded as 
	\begin{equation}
	d_{GH}(X^{l},\mathbb{X}) \leq d_{H}(\Pi_{V^{l}}(X),\mathbb{X}^{N})+d_{H}(\mathbb{X}^{N},\mathbb{X}).\label{eq:analysis_stability_gromov_hausdorff}
	\end{equation}
	Then from applying $\left\Vert \mathbb{X}_{i}^{N}-X_{i}\right\Vert _{2}\leq\sqrt{m}\left\Vert \epsilon\right\Vert _{\infty}$and
	$\left\Vert \mathbb{X}_{i}^{N}\right\Vert _{2}\leq\sqrt{m}L_{f}T$
	to Proposition \ref{prop:pca_stability_persistent_homology_linear},
	the first term of \eqref{eq:analysis_stability_gromov_hausdorff}
	is bounded as 
	\begin{align}
	d_{H}(\Pi_{V^{l}}(X),\mathbb{X}^{N}) & \leq\sqrt{m}\left\Vert \epsilon\right\Vert _{\infty}\left(1+\frac{2\sqrt{m}L_{f}T\left(2\sqrt{m}L_{f}T+\sqrt{m}\left\Vert \epsilon\right\Vert _{\infty}\right)}{\lambda_{\tilde{l}}}\right)\nonumber \\
	& \leq\sqrt{m}\left\Vert \epsilon\right\Vert _{\infty}+\frac{2m^{\frac{3}{2}}L_{f}T\left\Vert \epsilon\right\Vert _{\infty}(2L_{f}T+\left\Vert \epsilon\right\Vert _{\infty})}{\lambda_{\tilde{l}}}.\label{eq:analysis_stability_hausdorff_pca}
	\end{align}
	Now, consider the second term of \eqref{eq:analysis_stability_gromov_hausdorff},
	which is coming from the time sampling. Since we have assumed that
	$f$ is $L_{f}$-Lipschitz, i.e. $\left|f(t_{1})-f(t_{2})\right|\leq L_{f}|t_{1}-t_{2}|$,
	$SW_{m,\tau}f$ is also Lipschitz with constant $\sqrt{m}L_{f}$ as
	\begin{align*}
	\Vert SW_{m,\tau}f(t_{1})-SW_{m,\tau}(f)(t_{2})\Vert_{2} & =\sqrt{\sum_{k=0}^{m-1}\big|f(t_{1}+k(m-1)\tau)-f(t_{2}+k(m-1)\tau)\big|^{2}}\\
	& \leq\sqrt{m}L_{f}|t_{1}-t_{2}|.
	\end{align*}
	Since $\mathbb{X}_{N}$ is sampled from $\mathbb{X}$ by sampling
	grid size of $\frac{T}{N}$, $d_{H}(\mathbb{X}^{N},\mathbb{X})$ is
	bounded as 
	\begin{equation}
	d_{H}(\mathbb{X}^{N},\mathbb{X})\leq\frac{\sqrt{m}L_{f}T}{N}.\label{eq:analysis_stability_hausdorff_time}
	\end{equation}
	Hence by applying \eqref{eq:analysis_stability_hausdorff_pca} and
	\eqref{eq:analysis_stability_hausdorff_time} to \eqref{eq:analysis_stability_gromov_hausdorff},
	the Gromov-Hausdorff distance $d_{GH}(X^{l},\mathbb{X})$ can be bounded
	as 
	\begin{align}
	d_{GH}(X^{l},\mathbb{X}) & \leq d_{H}(\Pi_{V^{l}}(X),\mathbb{X}^{N})+d_{H}(\mathbb{X}^{N},\mathbb{X}).\nonumber \\
	& \leq\sqrt{m}\left\Vert \epsilon\right\Vert _{\infty}+\frac{2m^{\frac{3}{2}}L_{f}T\left\Vert \epsilon\right\Vert _{\infty}(2L_{f}T+\left\Vert \epsilon\right\Vert _{\infty})}{\lambda_{\tilde{l}}}+\frac{\sqrt{m}L_{f}T}{N}.\label{eq:analysis_stability_gromov_hausdorff_bound}
	\end{align}
	Now, from the assumption $\left\Vert \epsilon\right\Vert _{\infty}<\infty$,
	both $X^{l}$ and $\mathbb{X}$ are bounded subsets of Euclidean space, hence applying \eqref{eq:analysis_stability_gromov_hausdorff_bound}
	to Theorem \ref{thm:stability_rips}
	bounds the bottleneck distance $d_{B}(D_{X},D_{f})$ as 
	\begin{align*}
	d_{B}(D_{X},D_{f}) & \leq d_{GH}(X^{l},\mathbb{X})\\
	& \leq\sqrt{m}\left\Vert \epsilon\right\Vert _{\infty}+\frac{2m^{\frac{3}{2}}L_{f}T\left\Vert \epsilon\right\Vert _{\infty}(2L_{f}T+\left\Vert \epsilon\right\Vert _{\infty})}{\lambda_{\tilde{l}}}+\frac{\sqrt{m}L_{f}T}{N}.
	\end{align*}
	
\end{proof}

\begin{proof}[Proof of Theorem \ref{thm:analysis_stability_landscape_silhouette}]
	
	From Theorem A.1 in \cite{bubenik2015statistical} and Proposition
	\ref{prop:analysis_stability_persistent_homology}, the $l_\infty$ distance between the landscape functions is upper bounded as 
	\begin{align*}
	\left\Vert\lambda_{k,X}-\lambda_{k,f}\right\Vert_{\infty} & \leq d_{B}(D_{X},D_{f})\\
	& \leq\sqrt{m}\left\Vert \epsilon\right\Vert _{\infty}+\frac{2m^{\frac{3}{2}}L_{f}T\left\Vert \epsilon\right\Vert _{\infty}(2L_{f}T+\left\Vert \epsilon\right\Vert _{\infty})}{\lambda_{\tilde{l}}}+\frac{\sqrt{m}L_{f}T}{N}.
	\end{align*}
	And since the silhouette functions are weighted sums of $\Lambda_{p}$, the $l_\infty$ distance between them is also
	correspondingly bounded as 
	\begin{align*}
	\left\Vert\phi_{X}^{(p)}-\phi_{f}^{(p)}\right\Vert_{\infty} & \leq d_{B}(D_{X},D_{f})\\
	& \leq\sqrt{m}\left\Vert \epsilon\right\Vert _{\infty}+\frac{2m^{\frac{3}{2}}L_{f}T\left\Vert \epsilon\right\Vert _{\infty}(2L_{f}T+\left\Vert \epsilon\right\Vert _{\infty})}{\lambda_{\tilde{l}}}+\frac{\sqrt{m}L_{f}T}{N}.
	\end{align*}
	
\end{proof}

\end{document}